%% file: main.tex
\documentclass[11pt]{article}

\usepackage{fullpage}

\input{command-template}

\title{Fast and Simple $(1+\epsilon)\Delta$-Edge-Coloring of Dense Graphs}

\usepackage{authblk}

\author{Abhishek Dhawan\thanks{Email: adhawan2@illinois.edu. Partially supported by NSF RTG grant DMS-1937241.}}

\affil{Department of Mathematics, University of Illinois Urbana-Champaign}

\date{}

\begin{document}

\maketitle

\begin{abstract}
    Let $\epsilon \in (0, 1)$ and $n, \Delta \in \mathbb N$ be such that $\Delta = \Omega\left(\max\left\{\frac{\log n}{\epsilon},\, \left(\frac{1}{\epsilon}\log \frac{1}{\epsilon}\right)^2\right\}\right)$. Given an $n$-vertex $m$-edge simple graph $G$ of maximum degree $\Delta$, we present a randomized $O\left(m\,\log^3 \Delta\,/\,\epsilon^2\right)$-time algorithm that computes a proper $(1+\epsilon)\Delta$-edge-coloring of $G$ with high probability. This improves upon the best known results for a wide range of the parameters $\epsilon$, $n$, and $\Delta$. Our approach combines a \textit{flagging} strategy from earlier work of the author with a \textit{shifting} procedure employed by Duan, He, and Zhang for dynamic edge-coloring. The resulting algorithm is simple to implement and may be of practical interest.
\end{abstract}

\sloppy
\vspace{10pt}
\tableofcontents
\newpage

\input{intro}

\input{prelim}

\input{algorithm}
\input{proof}

\subsection*{Acknowledgements}

We thank Anton Bernshteyn for helpful discussions.
We are also grateful to the anonymous referees for carefully reading the manuscript and providing helpful suggestions.

\printbibliography

\end{document}

%% file: command-template.tex
\usepackage{makecell}
\usepackage{array}
\usepackage{comment}
\usepackage{epsf}
\usepackage{subcaption}
\usepackage{graphicx}
\usepackage{bbm}
\usepackage{color}
\usepackage{nicefrac}
\usepackage{tcolorbox}
\usepackage{mathtools}
\usepackage{amsfonts}
\usepackage{amsthm}
\usepackage{amsmath, bm}
\usepackage{amssymb}
\usepackage{textcomp}
\usepackage{lipsum}
\usepackage{xcolor}
\usepackage{graphicx}
\usepackage{wrapfig}
\usepackage{mathtools}
\usepackage{enumitem}
\usepackage{xspace}

\usepackage{algorithm}
\usepackage{algpseudocode}
\counterwithin{algorithm}{section}

\usepackage{tikz}
\usetikzlibrary{shapes,snakes}
\usetikzlibrary{arrows.meta}
\usetikzlibrary{decorations.pathmorphing}
\usetikzlibrary{patterns}
\usetikzlibrary{calc}

\usepackage{amsbsy}

\usepackage[linktocpage=true]{hyperref}
\definecolor{OliveGreen}{HTML}{3C8031}
\hypersetup{
colorlinks=true,
urlcolor=blue,
linkcolor=blue,
citecolor=OliveGreen,
}

\usepackage{xurl}
\hypersetup{breaklinks=true}

\newenvironment{fminipage}%
  {\end{minipage}\end{Sbox}\fbox{\TheSbox}}

\newtheorem{theorem}{Theorem}[section]

\newtheorem{proposition}[theorem]{Proposition}
\newtheorem*{proposition*}{Proposition}
\newtheorem{lemma}[theorem]{Lemma}
\newtheorem{corollary}[theorem]{Corollary}

\newtheorem{claim}{Claim}[theorem]

\theoremstyle{definition}
\newtheorem{definition}{Definition}[section]

\newenvironment{claimproof}{\noindent$\rhd$\hspace{1em}}{\hfill$\lhd$}

\makeatletter
\newenvironment{breakablealgorithm}
  {% \begin{breakablealgorithm}
   \begin{center}
     \refstepcounter{algorithm}% New algorithm
     \hrule height.8pt depth0pt \kern2pt% \@fs@pre for \@fs@ruled
     \renewcommand{\caption}[2][\relax]{% Make a new \caption
       {\raggedright\textbf{\ALG@name~\thealgorithm} ##2\par}%
       \ifx\relax##1\relax % #1 is \relax
         \addcontentsline{loa}{algorithm}{\protect\numberline{\thealgorithm}##2}%
       \else % #1 is not \relax
         \addcontentsline{loa}{algorithm}{\protect\numberline{\thealgorithm}##1}%
       \fi
       \kern2pt\hrule\kern2pt
     }
  }{% \end{breakablealgorithm}
     \kern2pt\hrule\relax% \@fs@post for \@fs@ruled
   \end{center}
  }
\makeatother

\usepackage[backend=biber, style=alphabetic, sorting=nyt, maxnames=100,backref=true, firstinits=true, sortcites = true]{biblatex}
\newlength{\widebarargwidth}
\newlength{\widebarargheight}
\newlength{\widebarargdepth}

\makeatletter
\long\def\@makecaption#1#2{
       \vskip 0.8ex
       \setbox\@tempboxa\hbox{\small {\bf #1:} #2}
       \parindent 1.5em 
       \dimen0=\hsize
       \advance\dimen0 by -3em
       \ifdim \wd\@tempboxa >\dimen0
               \hbox to \hsize{
                       \parindent 0em
                       \hfil 
                       \parbox{\dimen0}{\def\baselinestretch{0.96}\small
                               {\bf #1.} #2
                               } 
                       \hfil}
       \else \hbox to \hsize{\hfil \box\@tempboxa \hfil}
       \fi
       }
\makeatother

\long\def\comment#1{}

\newcommand{\bbone}[1]{\mathbbm{1}\{#1\}}
\newcommand{\set}[1]{\{#1\}}
\newcommand{\defeq}{\coloneqq}
\newcommand{\N}{\mathbb{N}}

\newcommand{\E}{\mathbb{E}}

\newcommand{\aug}{\mathsf{Aug}}
\newcommand{\flg}{\mathsf{flg}}

\newcommand{\poly}{\mathsf{poly}}

\newcommand{\vend}{\mathsf{vEnd}}
\newcommand{\vstart}{\mathsf{vStart}}

\newcommand{\End}{\mathsf{End}}
\newcommand{\Pivot}{\mathsf{Pivot}}
\newcommand{\length}{\mathsf{length}}
\renewcommand{\epsilon}{\varepsilon}
\newcommand{\eps}{\epsilon}

\newcommand{\dom}{\mathsf{dom}}
\newcommand{\0}{\varnothing}
\newcommand{\bemph}[1]{{\normalfont#1}} % define how emphasised brackets should look
\newcommand{\ep}[1]{\bemph{(}#1\bemph{)}} % parentheses

\newcommand{\emphdef}[1]{\textbf{\textit{{#1}}}}
\newcommand{\emphd}[1]{\emphdef{#1}}
\newcommand{\blank}{\mathsf{blank}}

%% file: intro.tex
\section{Introduction}\label{section:intro}

\subsection{Background and Results}\label{subsec:background}

All graphs considered in this paper are finite, undirected, and simple.
Let $G = (V, E)$ be a graph satisfying $|V| = n$, $|E| = m$, and $\Delta(G) = \Delta$, where $\Delta(G)$ is the maximum degree of a vertex in $G$.
For $q \in \N$, we let $[q] = \set{1, \ldots, q}$.
A proper $q$-edge-coloring of $G$ is a function $\phi\,:\, E \to [q]$ such that $\phi(e) \neq \phi(f)$ whenever $e$ and $f$ share an endpoint.
The chromatic index of $G$, denoted $\chi'(G)$, is the minimum value $q$ for which $G$ admits a proper $q$-edge-coloring.

It is easy to see that $\chi'(G) \geq \Delta$ as a vertex of maximum degree must have different colors assigned to all edges incident to it.
The following celebrated result of Vizing provides a nearly matching upper bound (see \cite[\S{}A.1]{EdgeColoringMonograph} for an English translation of the paper and \cites[\S17.2]{BondyMurty}[\S5.3]{Diestel} for modern presentations):

\begin{theorem}[Vizing's Theorem \cite{vizing1965chromatic}]\label{theo: vizing}
    If $G$ is a simple graph of maximum degree $\Delta$, then $\chi'(G) \leq \Delta + 1$.
\end{theorem}

In general, it is NP-hard to determine $\chi'(G)$ (even for $\Delta = 3$) \cite{Holyer}.
Therefore, it is natural to design algorithms for $q$-edge-coloring where $q \geq \Delta + 1$.
The original proof of Theorem~\ref{theo: vizing} is constructive and yields an $O(mn)$-time algorithm, simplifications of which appeared in \cite{Bollobas, RD, MG}.
Gabow, Nishizeki, Kariv, Leven, and Terada designed two recursive algorithms, which run in $O(m\sqrt{n\log n})$ and $O(\Delta\,m\log n)$ time, respectively \cite{GNKLT}.
There was no improvement for 34 years until Sinnamon described an algorithm with running time $O(m\sqrt{n})$ \cite{Sinnamon}, and more recently Bhattacharya, Carmon, Costa, Solomon, and Zhang designed an $\tilde O(mn^{1/3})$-time\footnote{Here, and in what follows, $\tilde O(x) = O(x\,\poly(\log x))$.} randomized algorithm \cite{bhattacharya2024faster}.

% \footnotetext{Here, and in what follows, $\tilde O(x) = O(x\,\poly(\log x))$}

Unsurprisingly, one can obtain faster algorithms for restricted classes of graphs.
Notably, Cole, Ost, and Schirra designed an $O(m\log \Delta)$-time algorithm for $\Delta$-edge-coloring bipartite graphs \cite{cole2001edge} (K\"onig's theorem states that $\chi'(G) = \Delta$ for bipartite graphs);
Bernshteyn and the author considered graphs of bounded maximum degree, describing a linear-time randomized algorithm for $(\Delta+1)$-edge-coloring in this setting, i.e., $O(m)$ for $\Delta = O(1)$ \cite{fastEdgeColoring};
Bhattacharya, Carmon, Costa, Solomon, and Zhang designed an $\tilde O(m)$-time algorithm for graphs of bounded arboricity \cite{bhattacharya2023density}.

Another avenue toward obtaining faster algorithms is to use larger palettes, i.e., more colors.
This is the main focus of this paper, where we study randomized algorithms for $(1 + \eps)\Delta$-edge-coloring.
We highlight recent results in Table~\ref{table:history}.

In our work, we design a randomized algorithm for dense graphs, i.e., 
\[\Delta = \Omega\left(\max\left\{\frac{\log n}{\eps},\, \left(\frac{1}{\eps}\log\frac{1}{\eps} \right)^2\right\}\right).\]
For a wide range of the parameters $\eps$, $n$, and $\Delta$, our algorithm outperforms those in Table~\ref{table:history}.
We discuss this further in \S\ref{subsection: prior}.
Let us first state our result.

\begin{table}[htb!]
    \centering
    \begin{tabular}{| c | c | c | p{0.2\linewidth}|}
        \hline
        \thead{\textbf{Runtime}} & \thead{\textbf{Restrictions on $\Delta,\,n,\,\eps$}} & \thead{\textbf{References}} \\\hline\hline
         $O\left(m\,\log^{c} n\right) $ for some large constant $c$ & $\Delta = \Theta\left(\log^2 n / \eps^2\right)$ & {Karloff--Shmoys \cite{karloff1987efficient}} \\[2pt]\hline
         $O\left(m\,\log^6 n/\epsilon^2\right) $ & $\Delta = \Omega\left(\log n / \eps\right)$ & {Duan--He--Zhang \cite{duan2019dynamic}} \\[2pt]\hline
          $O\left(m\log^3n/\eps^3\right) $ & $\Delta = \Omega\left(\log n / \eps\right)$ & {AD \cite{dhawan2024simple}} \\[2pt]\hline
          $O\left(\max \{m/\epsilon^{18}, \ m \,\log\Delta\}\right) $ & $\eps \geq \max\left\{\frac{\log \log n}{\log n}, \frac{1}{\Delta}\right\}$ & {Elkin--Khuzman \cite{elkin2024deterministic}} \\[4pt]\hline
          $O\left(m \,\log(1/\epsilon)/\epsilon^2\right) $ &  $\Delta \geq \left(\log n / \epsilon\right)^{\poly\left(1/\epsilon\right)}$ & \makecell{Bhattacharya--Costa--Panski--\\Solomon \cite{bhattacharya2024nibbling}} \\[2pt]\hline
          $O\left(m \,\log(1/\epsilon)\right) $ in expectation &  $\Delta = \Omega\left(\log n / \eps\right)$ & {Assadi \cite{assadi2024faster}} \\[2pt]\hline
          $O\left(m\log (1/\eps)/\eps^4\right) $ & No restrictions & {Bernshteyn--AD \cite{bernshteyn2024linear}} \\[2pt]\hline
    \end{tabular}
    
    \vspace{7pt}
    \caption{A brief survey of randomized algorithms for $(1+\eps)\Delta$-edge-coloring. The stated runtime is attained with high probability, unless explicitly indicated otherwise.}\label{table:history}
\end{table}

\begin{theorem}\label{theo:main_theo}
    Let $\epsilon \in (0,1)$ and $n,\Delta \in \N$ be such that $\Delta = \Omega\left(\max\left\{\frac{\log n}{\eps},\, \left(\frac{1}{\eps}\log\frac{1}{\eps} \right)^2\right\}\right)$.
    Let $G$ be an $n$-vertex graph with $m$ edges having maximum degree $\Delta$.
    There is an $O\left(m\log^3\Delta/\epsilon^2\right)$-time randomized sequential algorithm that finds a proper $(1+\epsilon)\Delta$-edge-coloring of $G$ with high probability\footnote{Throughout this paper, ``with high probability'' means with probability at least $1 - 1/\poly(n)$.}.
\end{theorem}

Throughout the paper, we assume that the hidden constant in the $\Omega(\cdot)$ notation above is sufficiently large.
The remainder of this introduction is structured as follows: in \S\ref{subsection: prior}, we compare Theorem~\ref{theo:main_theo} to earlier works; in \S\ref{subsection: proof overview}, we provide an informal overview of our algorithm; and in \S\ref{subsection: sketch}, we provide a sketch of the probabilistic analysis.

\subsection{Relation to Prior Work}\label{subsection: prior}

In this section, we provide a comparison of our result to those stated in Table~\ref{table:history}.
Clearly, Theorem~\ref{theo:main_theo} outperforms the earlier result of the author \cite{dhawan2024simple} and that of Duan, He, and Zhang \cite{duan2019dynamic}.
Bhattacharya, Costa, Panski, and Solomon's algorithm \cite{bhattacharya2024nibbling} is faster than ours, however, it covers a smaller range of parameters.
In particular, they require $\eps = \omega\left(\frac{\log \log \Delta}{\log \Delta}\right)$ and so their result does not hold for sparse palettes.
In relation to the result of Bernshteyn and the author \cite{bernshteyn2024linear}, our algorithm is faster when $\eps = O\left(\sqrt{\frac{\log \log \Delta}{\log^3\Delta}}\right)$.
Similarly, Theorem~\ref{theo:main_theo} improves upon the result of Elkin and Khuzman \cite{elkin2024deterministic} when $\eps = O\left(\frac{1}{\log^{3/16}\Delta}\right)$.
We remark that for each of the aforementioned bounds on $\eps$, there is a wide range of the parameters $n$ and $\Delta$ that satisfy the constraint of Theorem~\ref{theo:main_theo}.
Finally, we note that the result of Assadi holds only in expectation, whereas ours holds with high probability.
Through standard probabilistic boosting, we obtain the following corollary to Assadi's result:

\begin{corollary}[{Corollary to \cite[Theorem 4]{assadi2024faster}}]\label{corl: assadi}
    Let $\epsilon \in (0,1)$ and $n,\Delta \in \N$ be such that $\Delta = \Omega\left(\log n / \eps\right)$.
    Let $G$ be an $n$-vertex graph with $m$ edges having maximum degree $\Delta$.
    There is an $O\left(m\,\log n\, \log(1/\epsilon)\right)$-time randomized sequential algorithm that finds a proper $(1+\epsilon)\Delta$-edge-coloring of $G$ with high probability.
\end{corollary}

It can be verified that our result improves upon Corollary~\ref{corl: assadi} whenever
$\eps = \Omega\left(\left(\log\log\Delta\,\log\log\log \Delta\right)^{-1/2}\right)$ and $n = \exp\left(O\left(\left(\log n / \log \log n\right)^{1/3}\right)\right)$.
Note that this covers nearly all values of $\eps$, $n$, and $\Delta$ considered in Corollary~\ref{corl: assadi}.

Similarly, our result outperforms that of \cite{karloff1987efficient} whenever the following holds:
\[\Omega(\log^2n) = \Delta = o\left(\frac{\log^{c+2}n}{\log^3\log n}\right),\]
where $c$ is the constant in the bound on the running time of their algorithm (see Table~\ref{table:history}).

We remark that our result improves on the results of Table~\ref{table:history} for wider ranges of parameters, however, for ease of presentation, we state weaker bounds in this section.

\subsection{Informal Overview}\label{subsection: proof overview}

In this section, we will provide an informal overview of our algorithm.
For simplicity, we will describe how to compute a $(1 + O(\eps))\Delta$-edge-coloring (the argument for $q = (1+\eps)\Delta$ follows by re-parameterizing $\eps \gets \Theta(\eps)$).
Throughout the rest of the paper, we fix $\epsilon \in (0,1)$ and $n, \Delta \in \N$ such that $\Delta = \Omega\left(\max\left\{\frac{\log n}{\eps},\, \left(\frac{1}{\eps}\log\frac{1}{\eps} \right)^2\right\}\right)$ (where the hidden constant is assumed to be sufficiently large) and set $q \defeq (1+\epsilon)\Delta$. 
Without loss of generality, we may assume that $q$ is an integer.
We also fix a graph $G$ with $n$ vertices, $m$ edges, and maximum degree $\Delta$, and we write $V \defeq V(G)$ and $E \defeq E(G)$.
We call a function $\phi \colon E\to [q]\cup \{\blank\}$ a \emphd{partial $q$-edge-coloring} (or simply a \emphd{partial coloring}) of $G$. Here $\phi(e) = \blank$ indicates that the edge $e$ is uncolored. As usual, $\dom(\phi)$ denotes the \emphd{domain} of $\phi$, i.e., the set of all colored edges.

A standard approach toward edge-coloring is to construct so-called \emphd{augmenting subgraphs}.
The idea, in a nutshell, is to extend a partial coloring to include an uncolored edge by modifying the colors of ``few'' colored edges.
The subgraph induced by these edges is referred to as an augmenting subgraph.
Formally, we define augmenting subgraphs as follows:

\begin{definition}[{Augmenting subgraphs; \cite[Definition~1.4]{bernshteyn2024linear}}]\label{defn:aug}
    Let $\phi \colon E \to [q] \cup \set{\blank}$ be a proper partial $q$-edge-coloring with domain $\dom(\phi) \subset E$. A subgraph $H \subseteq G$ is \emphd{$e$-augmenting} for an uncolored edge $e \in E \setminus \dom(\phi)$ if $e \in E(H)$ and there is a proper partial coloring $\phi'$ with $\dom(\phi') = \dom(\phi) \cup \set{e}$ that agrees with $\phi$ on the edges that are not in $E(H)$; in other words, by only modifying the colors of the edges of $H$, it is possible to add $e$ to the set of colored edges. We refer to such a modification operation as \emphd{augmenting} $\phi$ using $H$.
\end{definition}

This yields an algorithmic framework for edge-coloring.
Namely, proceed in an iterative fashion and at each iteration, first construct an augmenting subgraph $H$ with respect to the current partial coloring $\phi$, and then augment $\phi$ using $H$.
This is the essence of the idea described earlier.
See Algorithm~\ref{temp:seq} for an outline of this framework.

% \vspace{0.1in}
{
\floatname{algorithm}{Algorithm Template}
\begin{algorithm}[htb!]\small
    \caption{A $q$-edge-coloring algorithm}\label{temp:seq}
    \begin{flushleft}
        \textbf{Input}: A graph $G = (V,E)$ of maximum degree $\Delta$. \\
        % \medskip
        \textbf{Output}: A proper $q$-edge-coloring of $G$.
    \end{flushleft}
    \begin{algorithmic}[1]
        \State $\phi \gets$ the empty coloring
        \While{there are uncolored edges}
            \State\label{step:choose_S} Pick an uncolored edge $e$.
            \State\label{step:construct_H} Find an $e$-augmenting subgraph $H$.
            \State Augment $\phi$ using $H$ (thus adding $e$ to the set of colored edges).
        \EndWhile
        \State \Return $\phi$
    \end{algorithmic}
\end{algorithm}
}

Nearly all of the algorithms in Table~\ref{table:history} employ this template (with some modifications).
In Vizing's original proof, he describes how to construct an $e$-augmenting subgraph (although he did not use this terminology) for any uncolored edge $e$ consisting of a \emphd{fan}---i.e., a set of edges incident to a common vertex---and an \emphd{alternating path}---i.e., a path whose edge colors form the sequence $\alpha$, $\beta$, $\alpha$, $\beta$, \ldots{} for some $\alpha$, $\beta \in [q]$; see Fig.~\ref{fig:vizing} for an illustration.
Such an augmenting subgraph, which we call a \emphd{Vizing chain} and denote $(F, P)$ for the fan $F$ and path $P$, can be constructed and augmented in time proportional to the lengths of $F$ and $P$, i.e., the number of edges contained in the Vizing chain.
As the length of a fan is trivially at most $\Delta$ and that of an alternating path is at most $n$, this yields the $O(m(\Delta + n)) = O(mn)$ runtime of the algorithms in \cite{Bollobas, RD, MG}.
There are now two challenges to overcome: long fans and long paths.

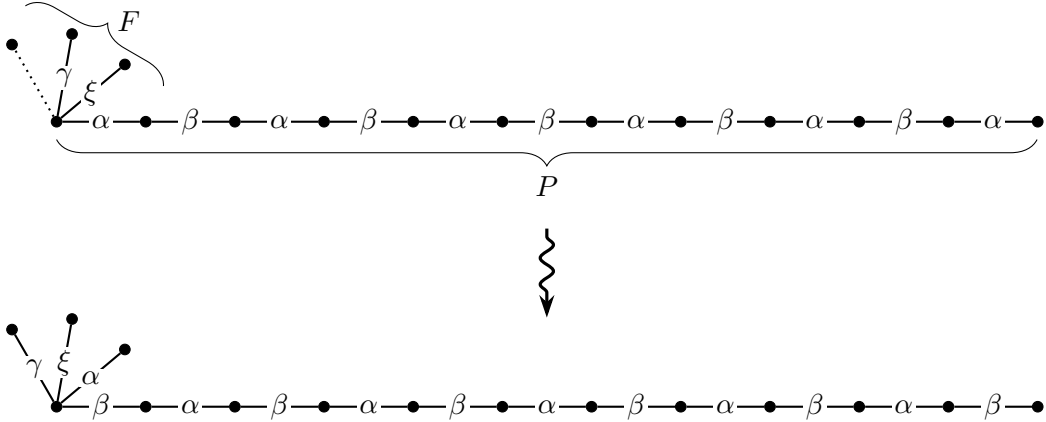
\begin{figure}[htb!]
    \centering
        \begin{tikzpicture}[scale=1.185]
            \node[circle,fill=black,draw,inner sep=0pt,minimum size=4pt] (a) at (0,0) {};
        	\path (a) ++(0:1) node[circle,fill=black,draw,inner sep=0pt,minimum size=4pt] (b) {};
        	\path (b) ++(0:1) node[circle,fill=black,draw,inner sep=0pt,minimum size=4pt] (c) {};
                \path (c) ++(0:1) node[circle,fill=black,draw,inner sep=0pt,minimum size=4pt] (d) {};
                \path (d) ++(0:1) node[circle,fill=black,draw,inner sep=0pt,minimum size=4pt] (e) {};
                \path (e) ++(0:1) node[circle,fill=black,draw,inner sep=0pt,minimum size=4pt] (l) {};
                \path (l) ++(0:1) node[circle,fill=black,draw,inner sep=0pt,minimum size=4pt] (m) {};
                \path (m) ++(0:1) node[circle,fill=black,draw,inner sep=0pt,minimum size=4pt] (n) {};
                \path (n) ++(0:1) node[circle,fill=black,draw,inner sep=0pt,minimum size=4pt] (o) {};
                \path (o) ++(0:1) node[circle,fill=black,draw,inner sep=0pt,minimum size=4pt] (f) {};
                \path (f) ++(0:1) node[circle,fill=black,draw,inner sep=0pt,minimum size=4pt] (g) {};
                \path (g) ++(0:1) node[circle,fill=black,draw,inner sep=0pt,minimum size=4pt] (h) {};

                \path (a) ++(40:1) node[circle,fill=black,draw,inner sep=0pt,minimum size=4pt] (i) {};
                \path (a) ++(80:1) node[circle,fill=black,draw,inner sep=0pt,minimum size=4pt] (j) {};
                \path (a) ++(120:1) node[circle,fill=black,draw,inner sep=0pt,minimum size=4pt] (k) {};

                \draw[thick] (a) to node[midway,inner sep=1pt,outer sep=1pt,minimum size=4pt,fill=white] {$\alpha$} (i) (a) to node[midway,inner sep=1pt,outer sep=1pt,minimum size=4pt,fill=white] {$\xi$} (j) (a) to node[midway,inner sep=1pt,outer sep=1pt,minimum size=4pt,fill=white] {$\gamma$} (k);
         
        	\draw[thick] (a) to node[midway,inner sep=1pt,outer sep=1pt,minimum size=4pt,fill=white] {$\beta$} (b) to node[midway,inner sep=1pt,outer sep=1pt,minimum size=4pt,fill=white] {$\alpha$} (c) to node[midway,inner sep=1pt,outer sep=1pt,minimum size=4pt,fill=white] {$\beta$} (d) to node[midway,inner sep=1pt,outer sep=1pt,minimum size=4pt,fill=white] {$\alpha$} (e) to node[midway,inner sep=1pt,outer sep=1pt,minimum size=4pt,fill=white] {$\beta$} (l) to node[midway,inner sep=1pt,outer sep=1pt,minimum size=4pt,fill=white] {$\alpha$} (m) to node[midway,inner sep=1pt,outer sep=1pt,minimum size=4pt,fill=white] {$\beta$} (n) to node[midway,inner sep=1pt,outer sep=1pt,minimum size=4pt,fill=white] {$\alpha$} (o) to node[midway,inner sep=1pt,outer sep=1pt,minimum size=4pt,fill=white] {$\beta$} (f) to node[midway,inner sep=1pt,outer sep=1pt,minimum size=4pt,fill=white] {$\alpha$} (g) to node[midway,inner sep=1pt,outer sep=1pt,minimum size=4pt,fill=white] {$\beta$} (h);

        \begin{scope}[yshift=3.2cm]
            \node[circle,fill=black,draw,inner sep=0pt,minimum size=4pt] (a) at (0,0) {};
        	\path (a) ++(0:1) node[circle,fill=black,draw,inner sep=0pt,minimum size=4pt] (b) {};
        	\path (b) ++(0:1) node[circle,fill=black,draw,inner sep=0pt,minimum size=4pt] (c) {};
                \path (c) ++(0:1) node[circle,fill=black,draw,inner sep=0pt,minimum size=4pt] (d) {};
                \path (d) ++(0:1) node[circle,fill=black,draw,inner sep=0pt,minimum size=4pt] (e) {};
                \path (e) ++(0:1) node[circle,fill=black,draw,inner sep=0pt,minimum size=4pt] (l) {};
                \path (l) ++(0:1) node[circle,fill=black,draw,inner sep=0pt,minimum size=4pt] (m) {};
                \path (m) ++(0:1) node[circle,fill=black,draw,inner sep=0pt,minimum size=4pt] (n) {};
                \path (n) ++(0:1) node[circle,fill=black,draw,inner sep=0pt,minimum size=4pt] (o) {};
                \path (o) ++(0:1) node[circle,fill=black,draw,inner sep=0pt,minimum size=4pt] (f) {};
                \path (f) ++(0:1) node[circle,fill=black,draw,inner sep=0pt,minimum size=4pt] (g) {};
                \path (g) ++(0:1) node[circle,fill=black,draw,inner sep=0pt,minimum size=4pt] (h) {};

                \path (a) ++(40:1) node[circle,fill=black,draw,inner sep=0pt,minimum size=4pt] (i) {};
                \path (a) ++(80:1) node[circle,fill=black,draw,inner sep=0pt,minimum size=4pt] (j) {};
                \path (a) ++(120:1) node[circle,fill=black,draw,inner sep=0pt,minimum size=4pt] (k) {};

                \draw[thick] (a) to node[midway,inner sep=1pt,outer sep=1pt,minimum size=4pt,fill=white] {$\xi$} (i) (a) to node[midway,inner sep=1pt,outer sep=1pt,minimum size=4pt,fill=white] {$\gamma$} (j);

                \draw[thick, dotted] (a) -- (k);
         
        	\draw[thick] (a) to node[midway,inner sep=1pt,outer sep=1pt,minimum size=4pt,fill=white] {$\alpha$} (b) to node[midway,inner sep=1pt,outer sep=1pt,minimum size=4pt,fill=white] {$\beta$} (c) to node[midway,inner sep=1pt,outer sep=1pt,minimum size=4pt,fill=white] {$\alpha$} (d) to node[midway,inner sep=1pt,outer sep=1pt,minimum size=4pt,fill=white] {$\beta$} (e) to node[midway,inner sep=1pt,outer sep=1pt,minimum size=4pt,fill=white] {$\alpha$} (l) to node[midway,inner sep=1pt,outer sep=1pt,minimum size=4pt,fill=white] {$\beta$} (m) to node[midway,inner sep=1pt,outer sep=1pt,minimum size=4pt,fill=white] {$\alpha$} (n) to node[midway,inner sep=1pt,outer sep=1pt,minimum size=4pt,fill=white] {$\beta$} (o) to node[midway,inner sep=1pt,outer sep=1pt,minimum size=4pt,fill=white] {$\alpha$} (f) to node[midway,inner sep=1pt,outer sep=1pt,minimum size=4pt,fill=white] {$\beta$} (g) to node[midway,inner sep=1pt,outer sep=1pt,minimum size=4pt,fill=white] {$\alpha$} (h);

            \draw[decoration={brace,amplitude=10pt,mirror}, decorate] (0, -0.2) -- node [midway,below,xshift=0pt,yshift=-10pt] {$P$} (11,-0.2);
            \draw[decoration={brace,amplitude=10pt},decorate] (-0.35,1.3) -- node [midway,above,yshift=2pt,xshift=13pt] {$F$} (1.2, 0.4);
            
        \end{scope}

        \begin{scope}[yshift=2.5cm]
            \draw[-{Stealth[length=3mm,width=2mm]},very thick,decoration = {snake,pre length=3pt,post length=7pt,},decorate] (5.5,-0.5) -- (5.5,-1.5);
        \end{scope}
        	
        \end{tikzpicture}
    \caption{The process of augmenting a Vizing chain $(F, P)$.}
    \label{fig:vizing}
\end{figure}

To deal with long fans, Duan, He, and Zhang had the following idea in \cite{duan2019dynamic}: sample a palette $C \subseteq [q]$ of size $\kappa$ and construct a Vizing chain $(F, P)$ using only the edges colored with colors from $C$ (in this case $\length(F) \leq \kappa + 1$).
For $\kappa = \Theta(\log n / \eps)$, they show that with high probability this procedure succeeds in constructing a Vizing chain $(F, P)$, and 
with probability $\Omega(1)$ the path $P$ is ``short''.
For an appropriate parameter $\ell$, they describe the following algorithm:
\begin{itemize}
    \item sample a palette $C \subseteq [q]$ satisfying $|C| = \kappa$, 
    \item construct an $e$-augmenting Vizing chain $(F, P)$ consisting of edges colored with colors from $C$,
    \item if $\length(P) > \ell$, try again.
\end{itemize}
For $\ell = \poly(\log n, 1/\eps)$, they are able to show that the above procedure succeeds within $\poly(\log n, 1/\eps)$ attempts with high probability.
In particular, they construct and augment an $e$-augmenting Vizing chain in $\poly(\log n, 1/\eps)$ time with high probability.

In recent work, the author improved upon the result of Duan, He, and Zhang by describing a two-stage coloring procedure \cite{dhawan2024simple}.
In Stage 1, color ``most'' of the edges, and in Stage 2, complete the coloring by employing a folklore $(2+\eps)\Delta$-edge-coloring algorithm (see Algorithm~\ref{alg:greedy}).
(We remark that Bhattacharya, Costa, Panski, and Solomon \cite{bhattacharya2024nibbling} and Assadi \cite{assadi2024faster} also employ a two-stage coloring procedure.)
Stage 1 of the author's algorithm in \cite{dhawan2024simple} follows similarly to Duan, He, and Zhang's algorithm, however, rather than repeating the procedure if the path $P$ has length $> \ell$, we pick one of the first $\ell$ edges on $P$ uniformly at random and \emphd{flag} it (flagged edges are colored during Stage 2).
We then augment the \emphd{initial segment} of the Vizing chain ending just before this edge.
(A similar idea was employed by Su and Vu in their distributed algorithm for $(\Delta + 2)$-edge-coloring \cite{su2019towards}.)
We illustrate this procedure in Fig.~\ref{fig:main_change} below.
A bulk of the proof in \cite{dhawan2024simple} involves showing that the flagged edges induce a graph of maximum degree $O(\eps\Delta)$ with high probability.
One can then put together the colorings from both stages to form a $(1+O(\eps))\Delta$-edge-coloring of $G$.

\begin{figure}[htb!]
    \centering
        \begin{tikzpicture}[scale=1.185]
            \node[circle,fill=black,draw,inner sep=0pt,minimum size=4pt] (a) at (0,0) {};
        	\path (a) ++(0:1) node[circle,fill=black,draw,inner sep=0pt,minimum size=4pt] (b) {};
        	\path (b) ++(0:1) node[circle,fill=black,draw,inner sep=0pt,minimum size=4pt] (c) {};
                \path (c) ++(0:1) node[circle,fill=black,draw,inner sep=0pt,minimum size=4pt] (d) {};
                \path (d) ++(0:1) node[circle,fill=black,draw,inner sep=0pt,minimum size=4pt] (e) {};
                \path (e) ++(0:1) node[circle,fill=black,draw,inner sep=0pt,minimum size=4pt] (l) {};
                \path (l) ++(0:1) node[circle,fill=black,draw,inner sep=0pt,minimum size=4pt] (m) {};
                \path (m) ++(0:1) node[circle,fill=black,draw,inner sep=0pt,minimum size=4pt] (n) {};
                \path (n) ++(0:1) node[circle,fill=black,draw,inner sep=0pt,minimum size=4pt] (o) {};
                \path (o) ++(0:1) node[circle,fill=black,draw,inner sep=0pt,minimum size=4pt] (f) {};
                \path (f) ++(0:1) node[circle,fill=black,draw,inner sep=0pt,minimum size=4pt] (g) {};
                \path (g) ++(0:1) node[circle,fill=black,draw,inner sep=0pt,minimum size=4pt] (h) {};

                \path (a) ++(40:1) node[circle,fill=black,draw,inner sep=0pt,minimum size=4pt] (i) {};
                \path (a) ++(80:1) node[circle,fill=black,draw,inner sep=0pt,minimum size=4pt] (j) {};
                \path (a) ++(120:1) node[circle,fill=black,draw,inner sep=0pt,minimum size=4pt] (k) {};

                \draw[thick] (a) to node[midway,inner sep=1pt,outer sep=1pt,minimum size=4pt,fill=white] {$\alpha$} (i) (a) to node[midway,inner sep=1pt,outer sep=1pt,minimum size=4pt,fill=white] {$\xi$} (j) (a) to node[midway,inner sep=1pt,outer sep=1pt,minimum size=4pt,fill=white] {$\gamma$} (k);
         
        	\draw[thick] (a) to node[midway,inner sep=1pt,outer sep=1pt,minimum size=4pt,fill=white] {$\beta$} (b) to node[midway,inner sep=1pt,outer sep=1pt,minimum size=4pt,fill=white] {$\alpha$} (c) to node[midway,inner sep=1pt,outer sep=1pt,minimum size=4pt,fill=white] {$\beta$} (d) to node[midway,inner sep=1pt,outer sep=1pt,minimum size=4pt,fill=white] {$\alpha$} (e) to node[midway,inner sep=1pt,outer sep=1pt,minimum size=4pt,fill=white] {$\beta$} (l) to node[midway,inner sep=1pt,outer sep=1pt,minimum size=4pt,fill=white] {$\alpha$} (m)  (n) to node[midway,inner sep=1pt,outer sep=1pt,minimum size=4pt,fill=white] {$\beta$} (o) to node[midway,inner sep=1pt,outer sep=1pt,minimum size=4pt,fill=white] {$\alpha$} (f) to node[midway,inner sep=1pt,outer sep=1pt,minimum size=4pt,fill=white] {$\beta$} (g) to node[midway,inner sep=1pt,outer sep=1pt,minimum size=4pt,fill=white] {$\alpha$} (h);

                \draw[thick, decorate, decoration=zigzag] (m) -- (n);
                \node at (6.5, 0.3) {flagged};

            \draw[decoration={brace,amplitude=10pt,mirror}, decorate] (0, -0.2) -- node [midway,below,xshift=0pt,yshift=-10pt] {$\ell'< \ell$} (6,-0.2);

        \begin{scope}[yshift=3.2cm]
            \node[circle,fill=black,draw,inner sep=0pt,minimum size=4pt] (a) at (0,0) {};
        	\path (a) ++(0:1) node[circle,fill=black,draw,inner sep=0pt,minimum size=4pt] (b) {};
        	\path (b) ++(0:1) node[circle,fill=black,draw,inner sep=0pt,minimum size=4pt] (c) {};
                \path (c) ++(0:1) node[circle,fill=black,draw,inner sep=0pt,minimum size=4pt] (d) {};
                \path (d) ++(0:1) node[circle,fill=black,draw,inner sep=0pt,minimum size=4pt] (e) {};
                \path (e) ++(0:1) node[circle,fill=black,draw,inner sep=0pt,minimum size=4pt] (l) {};
                \path (l) ++(0:1) node[circle,fill=black,draw,inner sep=0pt,minimum size=4pt] (m) {};
                \path (m) ++(0:1) node[circle,fill=black,draw,inner sep=0pt,minimum size=4pt] (n) {};
                \path (n) ++(0:1) node[circle,fill=black,draw,inner sep=0pt,minimum size=4pt] (o) {};
                \path (o) ++(0:1) node[circle,fill=black,draw,inner sep=0pt,minimum size=4pt] (f) {};
                \path (f) ++(0:1) node[circle,fill=black,draw,inner sep=0pt,minimum size=4pt] (g) {};
                \path (g) ++(0:1) node[circle,fill=black,draw,inner sep=0pt,minimum size=4pt] (h) {};

                \path (a) ++(40:1) node[circle,fill=black,draw,inner sep=0pt,minimum size=4pt] (i) {};
                \path (a) ++(80:1) node[circle,fill=black,draw,inner sep=0pt,minimum size=4pt] (j) {};
                \path (a) ++(120:1) node[circle,fill=black,draw,inner sep=0pt,minimum size=4pt] (k) {};

                \draw[thick] (a) to node[midway,inner sep=1pt,outer sep=1pt,minimum size=4pt,fill=white] {$\xi$} (i) (a) to node[midway,inner sep=1pt,outer sep=1pt,minimum size=4pt,fill=white] {$\gamma$} (j);

                \draw[thick, dotted] (a) -- (k);
         
        	\draw[thick] (a) to node[midway,inner sep=1pt,outer sep=1pt,minimum size=4pt,fill=white] {$\alpha$} (b) to node[midway,inner sep=1pt,outer sep=1pt,minimum size=4pt,fill=white] {$\beta$} (c) to node[midway,inner sep=1pt,outer sep=1pt,minimum size=4pt,fill=white] {$\alpha$} (d) to node[midway,inner sep=1pt,outer sep=1pt,minimum size=4pt,fill=white] {$\beta$} (e) to node[midway,inner sep=1pt,outer sep=1pt,minimum size=4pt,fill=white] {$\alpha$} (l) to node[midway,inner sep=1pt,outer sep=1pt,minimum size=4pt,fill=white] {$\beta$} (m) to node[midway,inner sep=1pt,outer sep=1pt,minimum size=4pt,fill=white] {$\alpha$} (n) to node[midway,inner sep=1pt,outer sep=1pt,minimum size=4pt,fill=white] {$\beta$} (o) to node[midway,inner sep=1pt,outer sep=1pt,minimum size=4pt,fill=white] {$\alpha$} (f) to node[midway,inner sep=1pt,outer sep=1pt,minimum size=4pt,fill=white] {$\beta$} (g) to node[midway,inner sep=1pt,outer sep=1pt,minimum size=4pt,fill=white] {$\alpha$} (h);

            \draw[decoration={brace,amplitude=10pt,mirror}, decorate] (0, -0.2) -- node [midway,below,xshift=0pt,yshift=-10pt] {$> \ell$} (11,-0.2);
            
        \end{scope}

        \begin{scope}[yshift=2.5cm]
            \draw[-{Stealth[length=3mm,width=2mm]},very thick,decoration = {snake,pre length=3pt,post length=7pt,},decorate] (5.5,-0.5) -- (5.5,-1.5);
        \end{scope}
        	
        \end{tikzpicture}
    \caption{The flagging procedure of \cite{dhawan2024simple}.}
    \label{fig:main_change}
\end{figure}
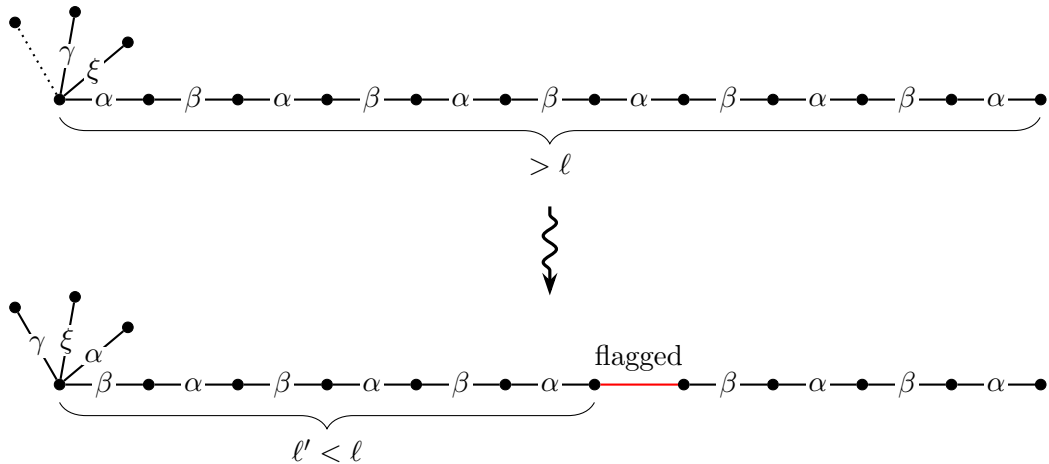

Our algorithm builds upon this flagging strategy.
Rather than flagging the sampled edge $f$, we \emphd{shift} the uncolored edge $e$ to $f$.
We do so by uncoloring $f$ and augmenting the corresponding initial segment of the Vizing chain; see Fig.~\ref{fig:shift} for an illustration.
We then attempt the same algorithm on this shifted uncolored edge with a new \emphd{disjoint palette}.
We continue in this fashion for at most $\Theta(\log \Delta)$ shifts.
The algorithm terminates if one of the following holds:
\begin{itemize}
    \item we fail to construct a Vizing chain from the sampled colors, or
    \item we run out of shifts.
\end{itemize}
In either case, we flag the uncolored edge and continue.
An algorithmic sketch of the first stage is provided in Algorithm~\ref{temp:stage1}.

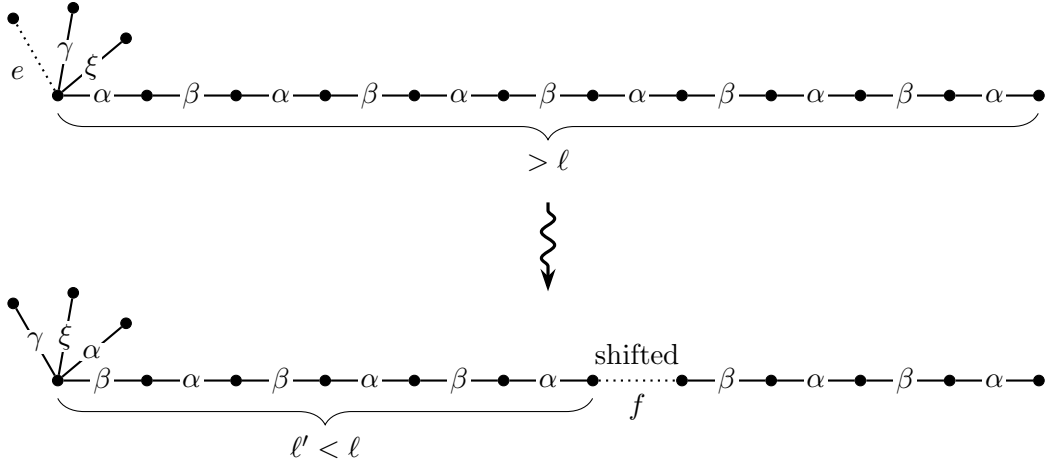
\begin{figure}[htb!]
    \centering
        \begin{tikzpicture}[scale=1.185]
            \node[circle,fill=black,draw,inner sep=0pt,minimum size=4pt] (a) at (0,0) {};
        	\path (a) ++(0:1) node[circle,fill=black,draw,inner sep=0pt,minimum size=4pt] (b) {};
        	\path (b) ++(0:1) node[circle,fill=black,draw,inner sep=0pt,minimum size=4pt] (c) {};
                \path (c) ++(0:1) node[circle,fill=black,draw,inner sep=0pt,minimum size=4pt] (d) {};
                \path (d) ++(0:1) node[circle,fill=black,draw,inner sep=0pt,minimum size=4pt] (e) {};
                \path (e) ++(0:1) node[circle,fill=black,draw,inner sep=0pt,minimum size=4pt] (l) {};
                \path (l) ++(0:1) node[circle,fill=black,draw,inner sep=0pt,minimum size=4pt] (m) {};
                \path (m) ++(0:1) node[circle,fill=black,draw,inner sep=0pt,minimum size=4pt] (n) {};
                \path (n) ++(0:1) node[circle,fill=black,draw,inner sep=0pt,minimum size=4pt] (o) {};
                \path (o) ++(0:1) node[circle,fill=black,draw,inner sep=0pt,minimum size=4pt] (f) {};
                \path (f) ++(0:1) node[circle,fill=black,draw,inner sep=0pt,minimum size=4pt] (g) {};
                \path (g) ++(0:1) node[circle,fill=black,draw,inner sep=0pt,minimum size=4pt] (h) {};

                \path (a) ++(40:1) node[circle,fill=black,draw,inner sep=0pt,minimum size=4pt] (i) {};
                \path (a) ++(80:1) node[circle,fill=black,draw,inner sep=0pt,minimum size=4pt] (j) {};
                \path (a) ++(120:1) node[circle,fill=black,draw,inner sep=0pt,minimum size=4pt] (k) {};

                \draw[thick] (a) to node[midway,inner sep=1pt,outer sep=1pt,minimum size=4pt,fill=white] {$\alpha$} (i) (a) to node[midway,inner sep=1pt,outer sep=1pt,minimum size=4pt,fill=white] {$\xi$} (j) (a) to node[midway,inner sep=1pt,outer sep=1pt,minimum size=4pt,fill=white] {$\gamma$} (k);
         
        	\draw[thick] (a) to node[midway,inner sep=1pt,outer sep=1pt,minimum size=4pt,fill=white] {$\beta$} (b) to node[midway,inner sep=1pt,outer sep=1pt,minimum size=4pt,fill=white] {$\alpha$} (c) to node[midway,inner sep=1pt,outer sep=1pt,minimum size=4pt,fill=white] {$\beta$} (d) to node[midway,inner sep=1pt,outer sep=1pt,minimum size=4pt,fill=white] {$\alpha$} (e) to node[midway,inner sep=1pt,outer sep=1pt,minimum size=4pt,fill=white] {$\beta$} (l) to node[midway,inner sep=1pt,outer sep=1pt,minimum size=4pt,fill=white] {$\alpha$} (m)  (n) to node[midway,inner sep=1pt,outer sep=1pt,minimum size=4pt,fill=white] {$\beta$} (o) to node[midway,inner sep=1pt,outer sep=1pt,minimum size=4pt,fill=white] {$\alpha$} (f) to node[midway,inner sep=1pt,outer sep=1pt,minimum size=4pt,fill=white] {$\beta$} (g) to node[midway,inner sep=1pt,outer sep=1pt,minimum size=4pt,fill=white] {$\alpha$} (h);

                \draw[thick, dotted] (m) -- (n);
                \node at (6.5, 0.3) {shifted};
                \node at (6.5, -0.3) {$f$};

            \draw[decoration={brace,amplitude=10pt,mirror}, decorate] (0, -0.2) -- node [midway,below,xshift=0pt,yshift=-10pt] {$\ell'< \ell$} (6,-0.2);

        \begin{scope}[yshift=3.2cm]
            \node[circle,fill=black,draw,inner sep=0pt,minimum size=4pt] (a) at (0,0) {};
        	\path (a) ++(0:1) node[circle,fill=black,draw,inner sep=0pt,minimum size=4pt] (b) {};
        	\path (b) ++(0:1) node[circle,fill=black,draw,inner sep=0pt,minimum size=4pt] (c) {};
                \path (c) ++(0:1) node[circle,fill=black,draw,inner sep=0pt,minimum size=4pt] (d) {};
                \path (d) ++(0:1) node[circle,fill=black,draw,inner sep=0pt,minimum size=4pt] (e) {};
                \path (e) ++(0:1) node[circle,fill=black,draw,inner sep=0pt,minimum size=4pt] (l) {};
                \path (l) ++(0:1) node[circle,fill=black,draw,inner sep=0pt,minimum size=4pt] (m) {};
                \path (m) ++(0:1) node[circle,fill=black,draw,inner sep=0pt,minimum size=4pt] (n) {};
                \path (n) ++(0:1) node[circle,fill=black,draw,inner sep=0pt,minimum size=4pt] (o) {};
                \path (o) ++(0:1) node[circle,fill=black,draw,inner sep=0pt,minimum size=4pt] (f) {};
                \path (f) ++(0:1) node[circle,fill=black,draw,inner sep=0pt,minimum size=4pt] (g) {};
                \path (g) ++(0:1) node[circle,fill=black,draw,inner sep=0pt,minimum size=4pt] (h) {};

                \path (a) ++(40:1) node[circle,fill=black,draw,inner sep=0pt,minimum size=4pt] (i) {};
                \path (a) ++(80:1) node[circle,fill=black,draw,inner sep=0pt,minimum size=4pt] (j) {};
                \path (a) ++(120:1) node[circle,fill=black,draw,inner sep=0pt,minimum size=4pt] (k) {};

                \draw[thick] (a) to node[midway,inner sep=1pt,outer sep=1pt,minimum size=4pt,fill=white] {$\xi$} (i) (a) to node[midway,inner sep=1pt,outer sep=1pt,minimum size=4pt,fill=white] {$\gamma$} (j);

                \draw[thick, dotted] (a) -- (k);
         
        	\draw[thick] (a) to node[midway,inner sep=1pt,outer sep=1pt,minimum size=4pt,fill=white] {$\alpha$} (b) to node[midway,inner sep=1pt,outer sep=1pt,minimum size=4pt,fill=white] {$\beta$} (c) to node[midway,inner sep=1pt,outer sep=1pt,minimum size=4pt,fill=white] {$\alpha$} (d) to node[midway,inner sep=1pt,outer sep=1pt,minimum size=4pt,fill=white] {$\beta$} (e) to node[midway,inner sep=1pt,outer sep=1pt,minimum size=4pt,fill=white] {$\alpha$} (l) to node[midway,inner sep=1pt,outer sep=1pt,minimum size=4pt,fill=white] {$\beta$} (m) to node[midway,inner sep=1pt,outer sep=1pt,minimum size=4pt,fill=white] {$\alpha$} (n) to node[midway,inner sep=1pt,outer sep=1pt,minimum size=4pt,fill=white] {$\beta$} (o) to node[midway,inner sep=1pt,outer sep=1pt,minimum size=4pt,fill=white] {$\alpha$} (f) to node[midway,inner sep=1pt,outer sep=1pt,minimum size=4pt,fill=white] {$\beta$} (g) to node[midway,inner sep=1pt,outer sep=1pt,minimum size=4pt,fill=white] {$\alpha$} (h);

            \draw[decoration={brace,amplitude=10pt,mirror}, decorate] (0, -0.2) -- node [midway,below,xshift=0pt,yshift=-10pt] {$> \ell$} (11,-0.2);
            \node at (-0.45, 0.25) {$e$};
            
        \end{scope}

        \begin{scope}[yshift=2.5cm]
            \draw[-{Stealth[length=3mm,width=2mm]},very thick,decoration = {snake,pre length=3pt,post length=7pt,},decorate] (5.5,-0.5) -- (5.5,-1.5);
        \end{scope}
        	
        \end{tikzpicture}
    \caption{The process of shifting an uncolored edge along a Vizing chain.}
    \label{fig:shift}
\end{figure}

% \vspace{0.1in}
{
\floatname{algorithm}{Algorithm Sketch}
\begin{algorithm}[htb!]\small
    \caption{Stage 1 of our $q$-edge-coloring algorithm}\label{temp:stage1}
    \begin{flushleft}
        \textbf{Input}: A graph $G = (V,E)$ of maximum degree $\Delta$. \\
        % \medskip
        \textbf{Output}: A partial $q$-edge-coloring of $G$.
    \end{flushleft}
    \begin{algorithmic}[1]
        \State $\phi \gets$ the empty coloring
        \While{there are uncolored and unflagged edges}
            \State $Q \gets [q]$.
            \State Pick an uncolored and unflagged edge $e$ uniformly at random. \label{step: random edge}
            \For{$t = 1, \ldots, \Theta(\log \Delta)$}
                \State\label{step: choose_C} Sample a palette $C \subseteq Q$ of size $\kappa$.
                \State $Q \gets Q \setminus C$.
                \State Construct an $e$-augmenting Vizing chain $(F, P)$ consisting of edges colored with colors from $C$.
                \If{the construction fails}
                    \State \textsf{Break}. \label{step:fail}
                \ElsIf{$\length(P) < \ell$}
                    \State Augment $\phi$ using $(F, P)$.
                    \State \textsf{Break}.
                \Else
                    \State Pick one of the first $\ell$ edges of $P$ uniformly at random and uncolor it. Call this edge $f$. \label{step: ell'}
                    \State Augment the corresponding \textit{initial segment} of the Vizing chain $(F, P)$.
                    \State $e \gets f$ \Comment{Shifting the uncolored edge.}
                \EndIf
            \EndFor
            \If{$e$ is uncolored}
                \State Flag $e$.
            \EndIf
        \EndWhile
        \State \Return $\phi$
    \end{algorithmic}
\end{algorithm}
}

We note that the idea of shifting uncolored edges across Vizing chains was first considered by Duan, He, and Zhang in their work on dynamic algorithms for $(1+\eps)\Delta$-edge-coloring \cite{duan2019dynamic}.
Additionally, it is closely related to so-called ``Multi-Step Vizing Chains'' introduced by Bernshteyn in \cite{VizingChain} building off of work of Greb\'ik and Pikhurko \cite{GP}, which have gathered much interest in recent years (see, e.g., \cite{fastEdgeColoring, Christ, dhawan2024edge, bernshteyn2023borel, grebik2023measurable}).

We conclude this section with a discussion on the simplicity of our algorithm with a comparison to those of Table~\ref{table:history}.
The algorithms of Elkin and Khuzman \cite{elkin2024deterministic} and Bernshteyn and the author \cite{bernshteyn2024linear} build upon earlier work of Bernshteyn and the author \cite{fastEdgeColoring}.
As mentioned in \S\ref{subsec:background}, it provides a $O(m)$-time $(\Delta + 1)$-edge-coloring algorithm for $\Delta = O(1)$.
The algorithm proceeds by iteratively employing their \textsf{Multi-Step Vizing Algorithm} to construct small augmenting multi-step Vizing chains of a special complex structure.
In \cite{elkin2024deterministic}, Elkin and Khuzman use it as a subroutine in order to construct a proper $(1+\epsilon)\Delta$-edge-coloring. 
Roughly, their idea is to recursively split the edges of $G$ into subgraphs of progressively smaller maximum degree, eventually reducing the degree to a constant, and then run the algorithm from \cite{fastEdgeColoring} on each constant-degree subgraph separately. 
The approach of \cite{bernshteyn2024linear} is to instead implement a variant of the \textsf{Multi-Step Vizing Algorithm} from \cite{fastEdgeColoring} on $G$ directly and make use of the larger set of available colors in the analysis.

As mentioned earlier, the algorithms of \cite{bhattacharya2024nibbling, assadi2024faster} also employ a two-stage coloring procedure.
Stage 2 of both algorithms is identical to ours (apply Algorithm~\ref{alg:greedy}).
For Stage 1, the authors of \cite{bhattacharya2024nibbling} employ the so-called ``R\"odl nibble method,'' an elaborate coloring procedure that has had great success in proving combinatorial results.
In \cite{assadi2024faster}, Assadi designs a procedure to construct $\Delta$ disjoint \emphd{fair matchings} in linear time.
The precise definition of ``fair'' is rather technical and unrelated to our work and so we omit it here.

Note that each of the ideas mentioned (multi-step Vizing chains, the nibble method, fair matchings, etc.) are complex and involve intricate arguments.
We emphasize that our approach is significantly simpler to implement (as well as analyze).

\subsection{Proof Sketch}\label{subsection: sketch}

In this section, we provide a sketch of the proof of Theorem~\ref{theo:main_theo}.
It should be understood that the presentation in this section deliberately ignores certain minor technical issues, and so the actual arguments in the rest of the paper may be slightly different from how they are described here. 
However, the differences do not affect the general conceptual framework.

For Stage 1, we implement Algorithm~\ref{temp:stage1} with parameters $\kappa = \Theta(\log \Delta / \eps)$ and $\ell = \Theta(\kappa^2)$.
The runtime of Stage 1 is then $O(m\log^3\Delta/\eps^2)$; Stage 2 can be implemented in $O(\max\set{m,\,\log^2n})$ time (see Proposition~\ref{prop:greedy}) and so the overall runtime is $O(m\log^3\Delta/\eps^2)$ as claimed.
It remains to show that the flagged edges induce a graph of maximum degree $O(\eps\Delta)$ with high probability.

Let us first discuss the probability of reaching Step~\ref{step:fail}.
Duan, He, and Zhang made the following observation regarding the success of their algorithm to construct a Vizing chain using sampled colors: as long as for every vertex $v$, there is at least one sampled color that is ``missing'' at $v$ (i.e., does not appear on an edge adjacent to $v$), the procedure succeeds.
For $\kappa = \Theta(\log n / \eps)$, this holds with high probability when sampling from $[q]$.
We note that it is not necessary for the above to hold for all $v$, but rather for $v \in N(x) \cup N(y)$, where $e = xy$ is the uncolored edge in question.
Furthermore, while we are not always sampling from $[q]$, we note that the set $Q$ at Step~\ref{step: choose_C} always satisfies the following as a result of our choice of $\kappa$ and the lower bound on $\Delta$ in Theorem~\ref{theo:main_theo}:
\[|Q| \geq q - \Theta(\kappa\,\log \Delta) \geq (1+\Omega(\eps))\Delta.\]
Therefore, for our choice of $\kappa$, we can show that the procedure fails with probability at most $1/\poly(\Delta)$.

Consider a vertex $v$.
Suppose that $v$ is incident to the flagged edge during the $i$-th iteration of the \textsf{while} loop in Algorithm~\ref{temp:stage1}.
Note that we flag an edge if we reach Step~\ref{step:fail} or we max out on iterations.
We provide a sketch for the former case (the latter case is similar).
To this end, we make a few definitions.
Suppose the \textsf{for} loop lasts $t$ iterations.
For each $1 \leq j < t$, define the following:
\begin{itemize}
    \item $C_j$: the palette sampled during the $j$-th iteration of the \textsf{for} loop,

    \item $(F_j, P_j)$: the Vizing chain computed during the $j$-th iteration of the \textsf{for} loop,

    \item $x_j$: the unique vertex incident to every edge on $F_j$, and

    \item $\alpha_j, \beta_j \in C_j$: the unique colors that appear on the edges of $P_j$.
\end{itemize}
Additionally, let $\phi$ be the partial coloring before entering the \textsf{for} loop and let $f = x_ty_t$ be the flagged edge (note that $v \in \set{x_t, y_t}$).
Without loss of generality, let $x_t$ be the vertex closer to $x_{t-1}$ on the path $P_{t-1}$.
As the palettes $C_1, \ldots, C_{t-1}$ are disjoint, we may conclude the following for $1 \leq j < t$:
\begin{itemize}
    \item the edges of $P_j$ are colored $\alpha_j$ and $\beta_j$ under $\phi$, and
    \item at least one of $\alpha_j$ or $\beta_j$ is missing at $x_j$ under $\phi$.
\end{itemize}
The second item above is a consequence of the way we construct Vizing chains.
Fig.~\ref{fig:example} depicts an example for $t = 4$.

\begin{figure}[htb!]
    \centering
        \begin{tikzpicture}[scale=1.185]
            \node[circle,fill=black,draw,inner sep=0pt,minimum size=4pt] (a) at (0,0) {};
        	\path (a) ++(0:1) node[circle,fill=black,draw,inner sep=0pt,minimum size=4pt] (b) {};
                \path (b) ++(0:1) node[circle,fill=black,draw,inner sep=0pt,minimum size=4pt] (c) {};
                \path (c) ++(0:1) node[circle,fill=black,draw,inner sep=0pt,minimum size=4pt] (d) {};
                \path (d) ++(0:1) node[circle,fill=black,draw,inner sep=0pt,minimum size=4pt] (e) {};
                \path (e) ++(0:1) node[circle,fill=black,draw,inner sep=0pt,minimum size=4pt] (f) {};
                \path (f) ++(0:1) node[circle,fill=black,draw,inner sep=0pt,minimum size=4pt] (g) {};
                \path (g) ++(0:1) node[circle,fill=black,draw,inner sep=0pt,minimum size=4pt] (h) {};
                \path (h) ++(0:1) node[circle,fill=black,draw,inner sep=0pt,minimum size=4pt] (i) {};
                \path (i) ++(0:1) node[circle,fill=black,draw,inner sep=0pt,minimum size=4pt] (j) {};
                \path (j) ++(0:1) node[circle,fill=black,draw,inner sep=0pt,minimum size=4pt] (k) {};
                \path (k) ++(0:1) node[circle,fill=black,draw,inner sep=0pt,minimum size=4pt] (l) {};
                \path (l) ++(0:1) node[circle,fill=black,draw,inner sep=0pt,minimum size=4pt] (m) {};
                \path (m) ++(0:1) node[circle,fill=black,draw,inner sep=0pt,minimum size=4pt] (n) {};
                
        	\draw[thick] (a) to node[midway,inner sep=1pt,outer sep=1pt,minimum size=4pt,fill=white] {$\alpha_1$} (b) to node[midway,inner sep=1pt,outer sep=1pt,minimum size=4pt,fill=white] {$\beta_1$} (c) to node[midway,inner sep=1pt,outer sep=1pt,minimum size=4pt,fill=white] {$\alpha_1$} (d) to node[midway,inner sep=1pt,outer sep=1pt,minimum size=4pt,fill=white] {$\beta_1$} (e) to node[midway,inner sep=1pt,outer sep=1pt,minimum size=4pt,fill=white] {$\alpha_2$} (f) to node[midway,inner sep=1pt,outer sep=1pt,minimum size=4pt,fill=white] {$\beta_2$} (g) to node[midway,inner sep=1pt,outer sep=1pt,minimum size=4pt,fill=white] {$\alpha_2$} (h) to node[midway,inner sep=1pt,outer sep=1pt,minimum size=4pt,fill=white] {$\beta_2$} (i) to node[midway,inner sep=1pt,outer sep=1pt,minimum size=4pt,fill=white] {$\alpha_3$} (j) to node[midway,inner sep=1pt,outer sep=1pt,minimum size=4pt,fill=white] {$\beta_3$} (k) to node[midway,inner sep=1pt,outer sep=1pt,minimum size=4pt,fill=white] {$\alpha_3$} (l) to node[midway,inner sep=1pt,outer sep=1pt,minimum size=4pt,fill=white] {$\beta_3$} (m) to node[midway,inner sep=1pt,outer sep=1pt,minimum size=4pt,fill=white] {$\alpha_3$} (n);

                \path (a) ++(40:1) node[circle,fill=black,draw,inner sep=0pt,minimum size=4pt] (a1) {};
                \path (a) ++(80:1) node[circle,fill=black,draw,inner sep=0pt,minimum size=4pt] (a2) {};
                \path (a) ++(120:1) node[circle,fill=black,draw,inner sep=0pt,minimum size=4pt] (a3) {};

                \path (e) ++(40:1) node[circle,fill=black,draw,inner sep=0pt,minimum size=4pt] (e1) {};
                \path (e) ++(80:1) node[circle,fill=black,draw,inner sep=0pt,minimum size=4pt] (e2) {};
                \path (e) ++(120:1) node[circle,fill=black,draw,inner sep=0pt,minimum size=4pt] (e3) {};

                \path (i) ++(40:1) node[circle,fill=black,draw,inner sep=0pt,minimum size=4pt] (i1) {};
                \path (i) ++(80:1) node[circle,fill=black,draw,inner sep=0pt,minimum size=4pt] (i2) {};
                \path (i) ++(120:1) node[circle,fill=black,draw,inner sep=0pt,minimum size=4pt] (i3) {};

                \draw[thick, dotted] (a) -- (a3);

                \draw[thick] (a) -- (a1) (a) -- (a2) (e) -- (e1) (e) -- (e2) (e) -- (e3) (i) -- (i1) (i) -- (i2) (i) -- (i3);

                \node at (-0.45, 0.25) {$e$};
                \node at (0, -0.25) {$x_1$};
                \node at (4, -0.25) {$x_2$};
                \node at (8, -0.25) {$x_3$};
                \node at (12, -0.25) {$x_4$};
                \node at (13, -0.25) {$y_4$};

                \draw[decoration={brace,amplitude=10pt},decorate] (-0.35,1.3) -- node [midway,above,yshift=2pt,xshift=13pt] {$F_1$} (1.2, 0.4);
                \draw[decoration={brace,amplitude=10pt},decorate] (3.65,1.3) -- node [midway,above,yshift=2pt,xshift=13pt] {$F_2$} (5.2, 0.4);
                \draw[decoration={brace,amplitude=10pt},decorate] (7.65,1.3) -- node [midway,above,yshift=2pt,xshift=13pt] {$F_3$} (9.2, 0.4);

                \draw[decoration={brace,amplitude=10pt},decorate] (12,0.1) -- node [midway,above,yshift=9pt,xshift=0pt] {$f$} (13, 0.1);
        \end{tikzpicture}
    \caption{An example of reaching Step~\ref{step:fail} when $t = 4$.}
    \label{fig:example}
\end{figure}
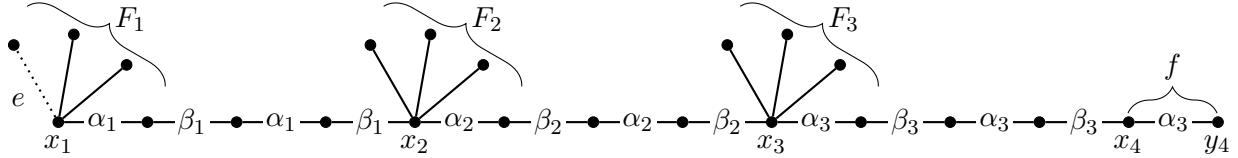

Given $v$, there are $O(\Delta)$ choices for the vertex $x_t$.
Furthermore, given $x_{j + 1}$, there are at most $|C_j|^2 = \kappa^2$ choices for $x_j$ as the colors $\alpha_j$ and $\beta_j$ uniquely determine $x_j$.
Finally, given $x_1$, there are at most $\Delta$ choices for the edge $e$.
It follows that the number of possible tuples $(e, x_1, \ldots, x_t)$ is at most $O(\Delta^2\kappa^{2(t-1)})$.

Note that during each iteration of the \textsf{while} loop, the number of uncolored and unflagged edges decreases by precisely $1$.
In particular, there are exactly $m - i + 1$ uncolored and unflagged edges at the start of the $i$-th iteration.
As the random choices made at Steps~\ref{step: random edge} and \ref{step: ell'} are independent, we may conclude that
\begin{align*}
    &~\Pr[v\text{ is incident to the flagged edge during the } i\text{-th iteration}] \\
    &\leq \sum_{t = 1}^{\Theta(\log \Delta)}O\left(\frac{\Delta^2}{(m-i+1)}\left(\frac{\kappa^2}{\ell}\right)^{t-1}\,\frac{1}{\poly(\Delta)}\right) \\
    &\leq \frac{1}{(m-i+1)\poly(\Delta)},
\end{align*}
for a sufficiently large hidden constant in the definition of $\ell$.
From here, we have
\[\E[\deg_{G'}(v)] \,=\, \sum_{i = 1}^m\frac{1}{(m-i+1)\poly(\Delta)} \,\leq\, \frac{\log n}{\poly(\Delta)}.\]
(Here we use that $1 + 1/2 + 1/3 + \cdots + 1/m \approx \log m$ and $m \leq n^2$.) 
Using a version of Azuma's inequality due to Kuszmaul and Qi \cite{azuma} (see Theorem~\ref{theo:azuma_supermartingale}), we are able to prove a concentration result showing that $\deg(v) = O(\eps\Delta)$ with high probability (here, we require $\Delta = \Omega(\log n / \eps)$).
The bound on $\Delta(G')$ then follows by a union bound over $V$.

The rest of the paper is structured as follows: in \S\ref{sec:prelim}, we introduce some terminology and background facts that will be used in our proofs; in \S\ref{sec:alg}, we will describe our algorithm in detail and prove its correctness; finally, in \S\ref{section: proof}, we will prove Theorem~\ref{theo:main_theo}.

%% file: prelim.tex
\section{Notation and Preliminaries}\label{sec:prelim}

The definitions and algorithm in this section are taken from \cite[\S2]{dhawan2024simple} and so the familiar reader may skip on to \S\ref{sec:alg}.
This section is split into three subsections. 
In the first, we will introduce some definitions and notation regarding Vizing chains.
In the second, we will describe the data structures we use to store the graph $G$ and attributes regarding a partial coloring $\phi$.
In the third, we will describe the folklore algorithm for $(2+\eps)\Delta$-edge-coloring, which will constitute Stage 2 of our main algorithm.
The astute reader may wonder why we employ this procedure as opposed to a greedy $(2\Delta - 1)$-edge-coloring algorithm.
The simple greedy $(2\Delta - 1)$-edge-coloring algorithm has running time $O(m\Delta)$, which is inefficient for our purposes.
One can obtain a $O(m\log \Delta)$-time algorithm by applying an intricate recursive degree-splitting procedure (see, e.g., \cite{Sinnamon, elkin2024deterministic}).
% which is sufficiently fast for our result.
While this running time is sufficient for our result,
% As part of the focus of this work is on simplicity, 
we defer to the $(2+\eps)\Delta$-edge-coloring algorithm as a part of the focus of this work is on simplicity.

\subsection{Vizing Chains}\label{subsec:chains}

For $q \in \N$, given a partial $q$-edge-coloring $\phi$ and $x \in V$, we let
\[M(\phi, x) \defeq [q]\setminus\{\phi(xy)\,:\, y \in N_G(x)\}\] 
be the set of all the \emphd{missing} colors at $x$ under the coloring $\phi$.
We note that $|M(\phi, x)| \geq \epsilon\Delta$ for $q = (1+\eps)\Delta$.
An uncolored edge $xy$ is \emphd{$\phi$-happy} if $M(\phi, x)\cap M(\phi, y)\neq \0$. 
If $e = xy$ is $\phi$-happy, we can extend the coloring $\phi$ by assigning any color in $M(\phi, x)\cap M(\phi, y)$ to $e$.

Given a proper partial coloring, we wish to modify it in order to create a partial coloring with a happy edge.
We will do so by constructing so-called called \emphd{Vizing chains} (see Definition~\ref{defn:viz}).
A Vizing chain consists of a \emphd{fan} and an \emphd{alternating path}.
Let us first describe fans.

\begin{definition}[{Fans; \cite[Definition 2.1]{dhawan2024simple}}]\label{defn:fans}
    A \emphd{fan} of length $k$ under a partial coloring $\phi$ is a sequence $F = (x, y_0, \ldots, y_{k-1})$ such that:
    \begin{itemize}
        \item $y_0$, \ldots, $y_{k-1}$ are distinct neighbors of $x$,
        \item $\phi(xy_0) = \blank$, and
        \item $\phi(xy_i) \in M(\phi,y_{i-1})$ for $1 \leq i < k$.
    \end{itemize}
    We refer to $x$ as the \emphd{pivot} of the fan, and let $\Pivot(F) \defeq x$, $\vstart(F) \defeq y_0$, and $\vend(F) \defeq y_{k-1}$ denote the pivot, start, and end vertices of a fan $F$. (This notation is uniquely determined unless $k = 1$.)
    Finally, we let $\length(F) = k$.
\end{definition}

Let $F = (x, y_0, \ldots, y_{k-1})$ be a fan under a proper partial coloring $\phi$.
Define the coloring $\psi$ as follows:
\[\psi(e) \defeq \left\{\begin{array}{cc}
    \phi(xy_{i+1}) & \text{if} \quad e = xy_i \quad \text{for} \quad 0 \leq i < k-1; \\
    \blank & \text{if} \quad e = xy_{k-1}; \\
    \phi(e) & \text{otherwise.}
\end{array}\right.\]
By \cite[Fact 2.1]{dhawan2024simple}, $\psi$ is a proper partial $q$-edge-coloring.
We say $\psi$ is obtained from $\phi$ by \emphd{shifting} the fan $F$.
Such a shifting procedure is described in Fig.~\ref{fig:fan}.

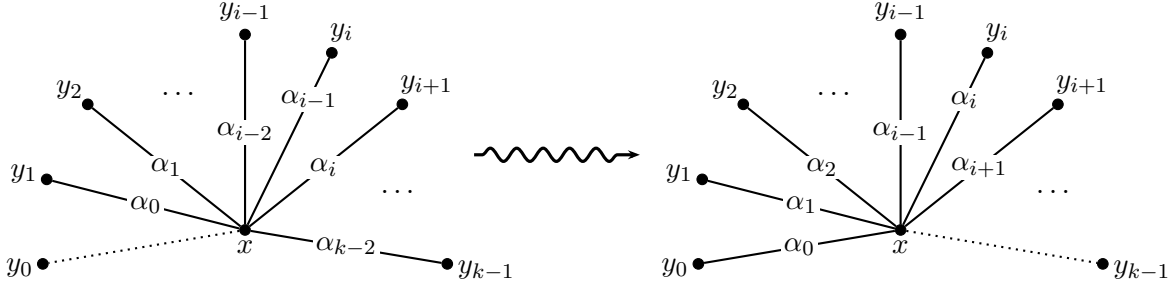
\begin{figure}[htb!]
	\centering
	\begin{tikzpicture}[xscale=1.05]
	\begin{scope}
	\node[circle,fill=black,draw,inner sep=0pt,minimum size=4pt] (x) at (0,0) {};
	\node[anchor=north] at (x) {$x$};
	
	\coordinate (O) at (0,0);
	\def\radius{2.6cm}
	
	\node[circle,fill=black,draw,inner sep=0pt,minimum size=4pt] (y0) at (190:\radius) {};
	\node at (190:2.9) {$y_0$};
	
	\node[circle,fill=black,draw,inner sep=0pt,minimum size=4pt] (y1) at (165:\radius) {};
	\node at (165:2.9) {$y_1$};
	
	\node[circle,fill=black,draw,inner sep=0pt,minimum size=4pt] (y2) at (140:\radius) {};
	\node at (140:2.9) {$y_2$};
	
	%\node[circle,fill=black,draw,inner sep=0pt,minimum size=4pt] (y3) at (115:\radius) {};
	\node[circle,fill=black,draw,inner sep=0pt,minimum size=4pt] (y4) at (90:\radius) {};
	\node at (90:2.9) {$y_{i-1}$};
	
	\node[circle,fill=black,draw,inner sep=0pt,minimum size=4pt] (y5) at (65:\radius) {};
	\node at (65:2.9) {$y_i$};
	
	\node[circle,fill=black,draw,inner sep=0pt,minimum size=4pt] (y6) at (40:\radius) {};
	\node at (40:3) {$y_{i+1}$};
	
	%\node[circle,fill=black,draw,inner sep=0pt,minimum size=4pt] (y7) at (15:\radius) {};
	\node[circle,fill=black,draw,inner sep=0pt,minimum size=4pt] (y8) at (-10:\radius) {};
	\node at (-10:3.1) {$y_{k-1}$};
	
	\node[circle,inner sep=0pt,minimum size=4pt] at (115:2) {$\ldots$}; 
	\node[circle,inner sep=0pt,minimum size=4pt] at (15:2) {$\ldots$}; 
	
	\draw[thick,dotted] (x) to (y0);
	\draw[thick] (x) to node[midway,inner sep=1pt,outer sep=1pt,minimum size=4pt,fill=white] {$\alpha_0$} (y1);
	\draw[thick] (x) to node[midway,inner sep=1pt,outer sep=1pt,minimum size=4pt,fill=white] {$\alpha_1$} (y2);
	
	\draw[thick] (x) to node[midway,inner sep=1pt,outer sep=1pt,minimum size=4pt,fill=white] {$\alpha_{i-2}$} (y4);
	\draw[thick] (x) to node[pos=0.75,inner sep=1pt,outer sep=1pt,minimum size=4pt,fill=white] {$\alpha_{i-1}$} (y5);
	\draw[thick] (x) to node[midway,inner sep=1pt,outer sep=1pt,minimum size=4pt,fill=white] {$\alpha_i$} (y6);
	
	\draw[thick] (x) to node[midway,inner sep=1pt,outer sep=1pt,minimum size=4pt,fill=white] {$\alpha_{k-2}$} (y8);
	\end{scope}
	
	\draw[-{Stealth[length=1.6mm]},very thick,decoration = {snake,pre length=3pt,post length=7pt,},decorate] (2.9,1) -- (5,1);
	
	\begin{scope}[xshift=8.3cm]
	\node[circle,fill=black,draw,inner sep=0pt,minimum size=4pt] (x) at (0,0) {};
	\node[anchor=north] at (x) {$x$};
	
	\coordinate (O) at (0,0);
	\def\radius{2.6cm}
	
	\node[circle,fill=black,draw,inner sep=0pt,minimum size=4pt] (y0) at (190:\radius) {};
	\node at (190:2.9) {$y_0$};
	
	\node[circle,fill=black,draw,inner sep=0pt,minimum size=4pt] (y1) at (165:\radius) {};
	\node at (165:2.9) {$y_1$};
	
	\node[circle,fill=black,draw,inner sep=0pt,minimum size=4pt] (y2) at (140:\radius) {};
	\node at (140:2.9) {$y_2$};
	
	%\node[circle,fill=black,draw,inner sep=0pt,minimum size=4pt] (y3) at (115:\radius) {};
	\node[circle,fill=black,draw,inner sep=0pt,minimum size=4pt] (y4) at (90:\radius) {};
	\node at (90:2.9) {$y_{i-1}$};
	
	\node[circle,fill=black,draw,inner sep=0pt,minimum size=4pt] (y5) at (65:\radius) {};
	\node at (65:2.9) {$y_i$};
	
	\node[circle,fill=black,draw,inner sep=0pt,minimum size=4pt] (y6) at (40:\radius) {};
	\node at (40:3) {$y_{i+1}$};
	
	%\node[circle,fill=black,draw,inner sep=0pt,minimum size=4pt] (y7) at (15:\radius) {};
	\node[circle,fill=black,draw,inner sep=0pt,minimum size=4pt] (y8) at (-10:\radius) {};
	\node at (-10:3.1) {$y_{k-1}$};
	
	\node[circle,inner sep=0pt,minimum size=4pt] at (115:2) {$\ldots$}; 
	\node[circle,inner sep=0pt,minimum size=4pt] at (15:2) {$\ldots$}; 
	
	\draw[thick] (x) to node[midway,inner sep=1pt,outer sep=1pt,minimum size=4pt,fill=white] {$\alpha_0$} (y0);
	\draw[thick] (x) to node[midway,inner sep=1pt,outer sep=1pt,minimum size=4pt,fill=white] {$\alpha_1$} (y1);
	\draw[thick] (x) to node[midway,inner sep=1pt,outer sep=1pt,minimum size=4pt,fill=white] {$\alpha_2$} (y2);
	
	\draw[thick] (x) to node[midway,inner sep=1pt,outer sep=1pt,minimum size=4pt,fill=white] {$\alpha_{i-1}$} (y4);
	\draw[thick] (x) to node[pos=0.75,inner sep=1pt,outer sep=1pt,minimum size=4pt,fill=white] {$\alpha_i$} (y5);
	\draw[thick] (x) to node[midway,inner sep=1pt,outer sep=1pt,minimum size=4pt,fill=white] {$\alpha_{i+1}$} (y6);
	
	\draw[thick, dotted] (x) to (y8);
	\end{scope}
	\end{tikzpicture}
	\caption{The process of shifting a fan.}\label{fig:fan}
\end{figure}

Next, we define alternating paths.

\begin{definition}[{Alternating Paths; \cite[Definition 2.2]{dhawan2024simple}}]\label{defn:path}
    For $\alpha, \beta \in [q]$, an \emphd{$\alpha\beta$-path} $P = (x_0, \ldots, x_k)$ under a partial $q$-edge-coloring $\phi$ is a sequence $P = (x_0, \ldots, x_k)$ such that:
    \begin{itemize}
        \item $x_0$, \ldots, $x_k$ are distinct vertices,
        \item $\phi(x_0x_1) = \alpha$, and
        \item the colors of the edges $x_ix_{i+1}$ alternate between $\alpha$ and $\beta$.
    \end{itemize}
    We let $\vstart(P) \defeq x_0$ and $\vend(P) \defeq x_k$ denote the first and last vertices on the path, respectively.
    Finally, we let $\length(P) = k$.
\end{definition}

Let $P = (x_0, \ldots, x_k)$ be a maximal $\alpha\beta$-path under a proper partial coloring $\phi$, i.e., $M(\phi, v) \cap \set{\alpha, \beta} \neq \0$ for $v \in \set{x_0, x_k}$.
Define the coloring $\psi$ by interchanging the colors of the edges in $P$.
By \cite[Fact 2.2]{dhawan2024simple}, $\psi$ is a proper partial $q$-edge-coloring.
We say $\psi$ is obtained from $\phi$ by \emphd{flipping} the path $P$.
Such a flipping procedure is described in Fig.~\ref{fig:path}.

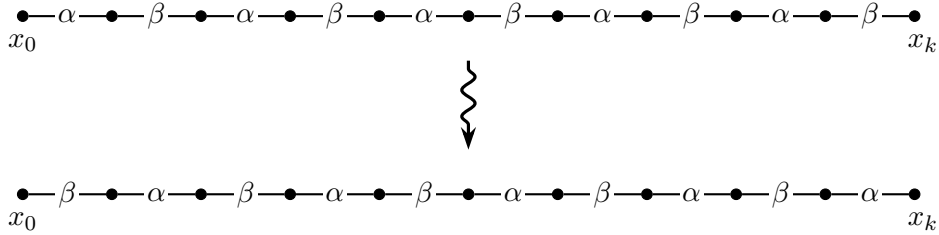
\begin{figure}[htb!]
    \centering
        \begin{tikzpicture}[scale=1.185]
            \node[circle,fill=black,draw,inner sep=0pt,minimum size=4pt] (a) at (0,0) {};
        	\path (a) ++(0:1) node[circle,fill=black,draw,inner sep=0pt,minimum size=4pt] (b) {};
        	\path (b) ++(0:1) node[circle,fill=black,draw,inner sep=0pt,minimum size=4pt] (c) {};
                \path (c) ++(0:1) node[circle,fill=black,draw,inner sep=0pt,minimum size=4pt] (d) {};
                \path (d) ++(0:1) node[circle,fill=black,draw,inner sep=0pt,minimum size=4pt] (e) {};
                \path (e) ++(0:1) node[circle,fill=black,draw,inner sep=0pt,minimum size=4pt] (f) {};
                \path (f) ++(0:1) node[circle,fill=black,draw,inner sep=0pt,minimum size=4pt] (g) {};
                \path (g) ++(0:1) node[circle,fill=black,draw,inner sep=0pt,minimum size=4pt] (h) {};
                \path (h) ++(0:1) node[circle,fill=black,draw,inner sep=0pt,minimum size=4pt] (i) {};
                \path (i) ++(0:1) node[circle,fill=black,draw,inner sep=0pt,minimum size=4pt] (j) {};
                \path (j) ++(0:1) node[circle,fill=black,draw,inner sep=0pt,minimum size=4pt] (k) {};
        	
        	% \draw[thick, decorate,decoration=zigzag] (e) to (f);

            \node at (0, -0.3) {$x_0$};
            \node at (10.1, -0.3) {$x_k$};
        	
        	\draw[thick] (a) to node[midway,inner sep=1pt,outer sep=1pt,minimum size=4pt,fill=white] {$\beta$} (b) to node[midway,inner sep=1pt,outer sep=1pt,minimum size=4pt,fill=white] {$\alpha$} (c) to node[midway,inner sep=1pt,outer sep=1pt,minimum size=4pt,fill=white] {$\beta$} (d) to node[midway,inner sep=1pt,outer sep=1pt,minimum size=4pt,fill=white] {$\alpha$} (e) to node[midway,inner sep=1pt,outer sep=1pt,minimum size=4pt,fill=white] {$\beta$} (f) to node[midway,inner sep=1pt,outer sep=1pt,minimum size=4pt,fill=white] {$\alpha$} (g) to node[midway,inner sep=1pt,outer sep=1pt,minimum size=4pt,fill=white] {$\beta$} (h) to node[midway,inner sep=1pt,outer sep=1pt,minimum size=4pt,fill=white] {$\alpha$} (i) to node[midway,inner sep=1pt,outer sep=1pt,minimum size=4pt,fill=white] {$\beta$} (j) to node[midway,inner sep=1pt,outer sep=1pt,minimum size=4pt,fill=white] {$\alpha$} (k);

        \begin{scope}[yshift=2cm]
            \node[circle,fill=black,draw,inner sep=0pt,minimum size=4pt] (a) at (0,0) {};
        	\path (a) ++(0:1) node[circle,fill=black,draw,inner sep=0pt,minimum size=4pt] (b) {};
        	\path (b) ++(0:1) node[circle,fill=black,draw,inner sep=0pt,minimum size=4pt] (c) {};
                \path (c) ++(0:1) node[circle,fill=black,draw,inner sep=0pt,minimum size=4pt] (d) {};
                \path (d) ++(0:1) node[circle,fill=black,draw,inner sep=0pt,minimum size=4pt] (e) {};
                \path (e) ++(0:1) node[circle,fill=black,draw,inner sep=0pt,minimum size=4pt] (f) {};
                \path (f) ++(0:1) node[circle,fill=black,draw,inner sep=0pt,minimum size=4pt] (g) {};
                \path (g) ++(0:1) node[circle,fill=black,draw,inner sep=0pt,minimum size=4pt] (h) {};
                \path (h) ++(0:1) node[circle,fill=black,draw,inner sep=0pt,minimum size=4pt] (i) {};
                \path (i) ++(0:1) node[circle,fill=black,draw,inner sep=0pt,minimum size=4pt] (j) {};
                \path (j) ++(0:1) node[circle,fill=black,draw,inner sep=0pt,minimum size=4pt] (k) {};
        	
        	% \draw[thick, decorate,decoration=zigzag] (e) to (f);

            \node at (0, -0.3) {$x_0$};
            \node at (10.1, -0.3) {$x_k$};
        	
        	\draw[thick] (a) to node[midway,inner sep=1pt,outer sep=1pt,minimum size=4pt,fill=white] {$\alpha$} (b) to node[midway,inner sep=1pt,outer sep=1pt,minimum size=4pt,fill=white] {$\beta$} (c) to node[midway,inner sep=1pt,outer sep=1pt,minimum size=4pt,fill=white] {$\alpha$} (d) to node[midway,inner sep=1pt,outer sep=1pt,minimum size=4pt,fill=white] {$\beta$} (e) to node[midway,inner sep=1pt,outer sep=1pt,minimum size=4pt,fill=white] {$\alpha$} (f) to node[midway,inner sep=1pt,outer sep=1pt,minimum size=4pt,fill=white] {$\beta$} (g) to node[midway,inner sep=1pt,outer sep=1pt,minimum size=4pt,fill=white] {$\alpha$} (h) to node[midway,inner sep=1pt,outer sep=1pt,minimum size=4pt,fill=white] {$\beta$} (i) to node[midway,inner sep=1pt,outer sep=1pt,minimum size=4pt,fill=white] {$\alpha$} (j) to node[midway,inner sep=1pt,outer sep=1pt,minimum size=4pt,fill=white] {$\beta$} (k);
        \end{scope}

        \begin{scope}[yshift=1.5cm]
            \draw[-{Stealth[length=3mm,width=2mm]},very thick,decoration = {snake,pre length=3pt,post length=7pt,},decorate] (5,0) -- (5,-1);
        \end{scope}
        	
        \end{tikzpicture}
    \caption{The process of flipping an $\alpha\beta$-path.}
    \label{fig:path}
\end{figure}

Combining a fan and an alternating path yields a Vizing chain.

\begin{definition}[{Vizing chains; \cite[Definition 2.3]{dhawan2024simple}}]\label{defn:viz}
    A \emphd{Vizing chain} in a partial $q$-edge-coloring $\phi$ is a tuple $(F, P)$ such that $F = (x, y_0, \ldots, y_{k-1})$ is a fan and $P = (x_0 = x, x_1, \ldots, x_s)$ is an $\alpha\beta$-path for some $\alpha, \beta \in [q]$.
\end{definition}

A key part of our algorithm will involve constructing ``short'' Vizing chains $(F, P)$ and modifying our coloring by flipping $P$ and shifting $F$ to create a happy edge.
Such a procedure is described in Fig.~\ref{fig:vizing_shift}.

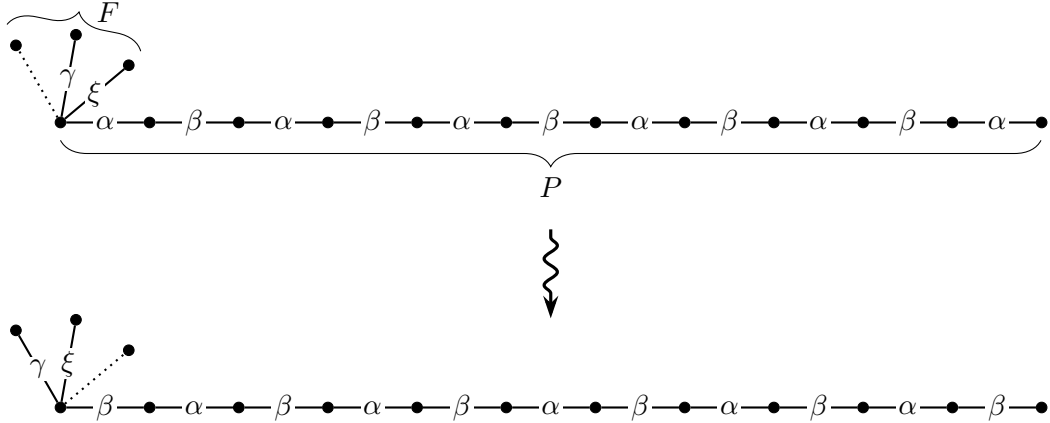
\begin{figure}[htb!]
    \centering
        \begin{tikzpicture}[scale=1.185]
            \node[circle,fill=black,draw,inner sep=0pt,minimum size=4pt] (a) at (0,0) {};
        	\path (a) ++(0:1) node[circle,fill=black,draw,inner sep=0pt,minimum size=4pt] (b) {};
        	\path (b) ++(0:1) node[circle,fill=black,draw,inner sep=0pt,minimum size=4pt] (c) {};
                \path (c) ++(0:1) node[circle,fill=black,draw,inner sep=0pt,minimum size=4pt] (d) {};
                \path (d) ++(0:1) node[circle,fill=black,draw,inner sep=0pt,minimum size=4pt] (e) {};
                \path (e) ++(0:1) node[circle,fill=black,draw,inner sep=0pt,minimum size=4pt] (l) {};
                \path (l) ++(0:1) node[circle,fill=black,draw,inner sep=0pt,minimum size=4pt] (m) {};
                \path (m) ++(0:1) node[circle,fill=black,draw,inner sep=0pt,minimum size=4pt] (n) {};
                \path (n) ++(0:1) node[circle,fill=black,draw,inner sep=0pt,minimum size=4pt] (o) {};
                \path (o) ++(0:1) node[circle,fill=black,draw,inner sep=0pt,minimum size=4pt] (f) {};
                \path (f) ++(0:1) node[circle,fill=black,draw,inner sep=0pt,minimum size=4pt] (g) {};
                \path (g) ++(0:1) node[circle,fill=black,draw,inner sep=0pt,minimum size=4pt] (h) {};

                \path (a) ++(40:1) node[circle,fill=black,draw,inner sep=0pt,minimum size=4pt] (i) {};
                \path (a) ++(80:1) node[circle,fill=black,draw,inner sep=0pt,minimum size=4pt] (j) {};
                \path (a) ++(120:1) node[circle,fill=black,draw,inner sep=0pt,minimum size=4pt] (k) {};
        	
                \draw[thick, dotted] (a) -- (i);
                
                \draw[thick] (a) to node[midway,inner sep=1pt,outer sep=1pt,minimum size=4pt,fill=white] {$\xi$} (j) (a) to node[midway,inner sep=1pt,outer sep=1pt,minimum size=4pt,fill=white] {$\gamma$} (k);
         
        	\draw[thick] (a) to node[midway,inner sep=1pt,outer sep=1pt,minimum size=4pt,fill=white] {$\beta$} (b) to node[midway,inner sep=1pt,outer sep=1pt,minimum size=4pt,fill=white] {$\alpha$} (c) to node[midway,inner sep=1pt,outer sep=1pt,minimum size=4pt,fill=white] {$\beta$} (d) to node[midway,inner sep=1pt,outer sep=1pt,minimum size=4pt,fill=white] {$\alpha$} (e) to node[midway,inner sep=1pt,outer sep=1pt,minimum size=4pt,fill=white] {$\beta$} (l) to node[midway,inner sep=1pt,outer sep=1pt,minimum size=4pt,fill=white] {$\alpha$} (m) to node[midway,inner sep=1pt,outer sep=1pt,minimum size=4pt,fill=white] {$\beta$} (n) to node[midway,inner sep=1pt,outer sep=1pt,minimum size=4pt,fill=white] {$\alpha$} (o) to node[midway,inner sep=1pt,outer sep=1pt,minimum size=4pt,fill=white] {$\beta$} (f) to node[midway,inner sep=1pt,outer sep=1pt,minimum size=4pt,fill=white] {$\alpha$} (g) to node[midway,inner sep=1pt,outer sep=1pt,minimum size=4pt,fill=white] {$\beta$} (h);

        \begin{scope}[yshift=3.2cm]
            \node[circle,fill=black,draw,inner sep=0pt,minimum size=4pt] (a) at (0,0) {};
        	\path (a) ++(0:1) node[circle,fill=black,draw,inner sep=0pt,minimum size=4pt] (b) {};
        	\path (b) ++(0:1) node[circle,fill=black,draw,inner sep=0pt,minimum size=4pt] (c) {};
                \path (c) ++(0:1) node[circle,fill=black,draw,inner sep=0pt,minimum size=4pt] (d) {};
                \path (d) ++(0:1) node[circle,fill=black,draw,inner sep=0pt,minimum size=4pt] (e) {};
                \path (e) ++(0:1) node[circle,fill=black,draw,inner sep=0pt,minimum size=4pt] (l) {};
                \path (l) ++(0:1) node[circle,fill=black,draw,inner sep=0pt,minimum size=4pt] (m) {};
                \path (m) ++(0:1) node[circle,fill=black,draw,inner sep=0pt,minimum size=4pt] (n) {};
                \path (n) ++(0:1) node[circle,fill=black,draw,inner sep=0pt,minimum size=4pt] (o) {};
                \path (o) ++(0:1) node[circle,fill=black,draw,inner sep=0pt,minimum size=4pt] (f) {};
                \path (f) ++(0:1) node[circle,fill=black,draw,inner sep=0pt,minimum size=4pt] (g) {};
                \path (g) ++(0:1) node[circle,fill=black,draw,inner sep=0pt,minimum size=4pt] (h) {};

                \path (a) ++(40:1) node[circle,fill=black,draw,inner sep=0pt,minimum size=4pt] (i) {};
                \path (a) ++(80:1) node[circle,fill=black,draw,inner sep=0pt,minimum size=4pt] (j) {};
                \path (a) ++(120:1) node[circle,fill=black,draw,inner sep=0pt,minimum size=4pt] (k) {};

                \draw[thick] (a) to node[midway,inner sep=1pt,outer sep=1pt,minimum size=4pt,fill=white] {$\xi$} (i) (a) to node[midway,inner sep=1pt,outer sep=1pt,minimum size=4pt,fill=white] {$\gamma$} (j);

                \draw[thick, dotted] (a) -- (k);
         
        	\draw[thick] (a) to node[midway,inner sep=1pt,outer sep=1pt,minimum size=4pt,fill=white] {$\alpha$} (b) to node[midway,inner sep=1pt,outer sep=1pt,minimum size=4pt,fill=white] {$\beta$} (c) to node[midway,inner sep=1pt,outer sep=1pt,minimum size=4pt,fill=white] {$\alpha$} (d) to node[midway,inner sep=1pt,outer sep=1pt,minimum size=4pt,fill=white] {$\beta$} (e) to node[midway,inner sep=1pt,outer sep=1pt,minimum size=4pt,fill=white] {$\alpha$} (l) to node[midway,inner sep=1pt,outer sep=1pt,minimum size=4pt,fill=white] {$\beta$} (m) to node[midway,inner sep=1pt,outer sep=1pt,minimum size=4pt,fill=white] {$\alpha$} (n) to node[midway,inner sep=1pt,outer sep=1pt,minimum size=4pt,fill=white] {$\beta$} (o) to node[midway,inner sep=1pt,outer sep=1pt,minimum size=4pt,fill=white] {$\alpha$} (f) to node[midway,inner sep=1pt,outer sep=1pt,minimum size=4pt,fill=white] {$\beta$} (g) to node[midway,inner sep=1pt,outer sep=1pt,minimum size=4pt,fill=white] {$\alpha$} (h);

            \draw[decoration={brace,amplitude=10pt,mirror}, decorate] (0, -0.2) -- node [midway,below,xshift=0pt,yshift=-10pt] {$P$} (11,-0.2);
            \draw[decoration={brace,amplitude=10pt},decorate] (-0.6,1.1) -- node [midway,above,yshift=2pt,xshift=13pt] {$F$} (0.9, 0.8);
            
        \end{scope}

        \begin{scope}[yshift=2.5cm]
            \draw[-{Stealth[length=3mm,width=2mm]},very thick,decoration = {snake,pre length=3pt,post length=7pt,},decorate] (5.5,-0.5) -- (5.5,-1.5);
        \end{scope}
        	
        \end{tikzpicture}
    \caption{The process of flipping $P$ and then shifting $F$ in a Vizing chain $(F, P)$. The edge $\End(F)$ can be colored $\alpha$.}
    \label{fig:vizing_shift}
\end{figure}

\subsection{Data Structures}\label{subsec:data_structures}
    
In this section, we describe how we will store our graph $G$ and partial coloring $\phi$.
Below, we discuss how our choices for the data structures affect the runtime of certain procedures.

\begin{itemize}
    \item We store $G$ as a list of vertices and edges, and include the partial coloring $\phi$ as an attribute of the graph.

    \item We store the partial coloring $\phi$ as a hash map, which maps edges to their respective colors; we map the edge to $\blank$ or $\flg$ if the edge is uncolored or flagged, respectively.
    Furthermore, the missing sets $M(\phi, \cdot)$ are also stored as hash maps, which map a vertex $x$ to a $q$-element array such that the following holds for each $\alpha \in [q]$:
    \[M(\phi, x)[\alpha] = \left\{\begin{array}{cc}
        y & \text{such that $\phi(xy) = \alpha$;} \\
        \blank & \text{if no such $y$ exists.}
    \end{array}\right.\]
    Note that as $\phi$ is a proper partial coloring, the vertex $y$ above is unique.
    In the remainder of the paper, we will use the notation $M(\phi, x)[\cdot]$ as described above in our algorithms, and the notation $M(\phi, x)$ to indicate the set of missing colors at $x$ in our proofs.
\end{itemize}
By our choice of data structures, we may shift a fan $F$ (resp. flip a path $P$) in time $O(\length(F))$ (resp. $O(\length(P))$).

\subsection{An Algorithm for Stage 2}\label{subsec:stage two}

In this section, we describe the folklore algorithm for $(2+\eps)\Delta$-edge-coloring, which we will invoke during Stage 2 of our main algorithm.
The algorithm proceeds as follows: iterate over the edges in an arbitrary order and at each iteration, repeatedly pick a color from $[(2+\eps)\Delta]$ uniformly at random until a valid one is selected (a color $\alpha$ is valid for an edge $e$ if no edge sharing an endpoint with $e$ is colored $\alpha$).

% \vspace{5pt}
\begin{algorithm}[htb!]\small
\caption{$(2+\eps)\Delta$-Edge-Coloring}\label{alg:greedy}
\begin{flushleft}
\textbf{Input}: An $n$-vertex graph $G$ having $m$ edges and maximum degree $\Delta$. \\
% \medskip
\textbf{Output}: A proper $(2 + \epsilon)\Delta$-edge-coloring of $G$.
\end{flushleft}

\begin{algorithmic}[1]
    \State $\phi(e) \gets \blank$ \textbf{for each} $e \in E$
    \For{$e = xy \in E$}
        \State \label{step:random color choice} Pick $\alpha \in [(2+\eps)\Delta]$ uniformly at random.
        \If{$M(\phi, x)[\alpha] = \blank$ and $M(\phi, y)[\alpha] = \blank$}
            \State $\phi(e) \gets \alpha$
        \Else
            \State Return to Step~\ref{step:random color choice}.
        \EndIf
    \EndFor
    \State \Return $\phi$
\end{algorithmic}
\end{algorithm}
% \vspace{5pt}

The following proposition provides a bound on the runtime of Algorithm~\ref{alg:greedy}.

\begin{proposition}[{\cite[Proposition 2.3]{dhawan2024simple}}]\label{prop:greedy}
    Let $\eps > 0$ be arbitrary and let $\gamma \defeq \min\set{1,\, \eps}$.
    For any $n$-vertex graph $G$ with $m$ edges and maximum degree $\Delta$,
    Algorithm~\ref{alg:greedy} computes a proper $(2+\eps)\Delta$-edge-coloring of $G$ in $O(\max\set{m/\gamma^2,\,\log^2n/\gamma})$ time with probability at least $1 - \min\set{e^{-m/\gamma}, 1/\poly(n)}$.
\end{proposition}

%% file: algorithm.tex
\section{Algorithm Overview}\label{sec:alg}

Our algorithm consists of two stages:
in Stage 1, we will find a proper partial $(1+\eps/2)\Delta$-edge-coloring of $G$;
in Stage 2, we will invoke Algorithm~\ref{alg:greedy} to color the remaining edges with $\eps\Delta/2$ colors.
Below, we describe our algorithm modulo a technical subprocedure, which we defer to \S\ref{subsec: color one overview}.

To assist with the description, we let $\flg(\phi)$ denote the flagged edges under a partial coloring $\phi$.
We begin with the blank coloring and proceed iteratively.
At each iteration, we pick an uncolored and unflagged edge $e$ and a vertex $x \in e$ uniformly at random and invoke the algorithm \hyperref[alg:color]{Color One} (see Algorithm~\ref{alg:color}) with input $(\phi, e, x, \kappa, \ell)$, where $\phi$ is the current partial coloring and $\kappa$ and $\ell$ are two constants to be chosen later.
The precise description of \hyperref[alg:color]{Color One} is rather technical and so we defer it to \S\ref{subsec: color one overview}.
For now, we may assume that it returns a partial coloring $\psi$ such that
\begin{itemize}
    \item either $\dom(\psi) = \dom(\phi) \cup \set{e}$ and $\flg(\psi) = \flg(\phi)$, or
    \item there is an edge $f \in \dom(\phi) \cup \set{e}$ such that $\dom(\psi) = (\dom(\phi) \cup \set{e}) \setminus \set{f}$ and $\flg(\psi) = \flg(\phi) \cup \set{f}$.
\end{itemize}
We continue iterating until all edges are either colored or flagged.
This concludes Stage 1.
Let $\phi$ be the partial coloring from the first stage.
Let $G^*$ be the subgraph induced by the flagged edges.
If $\Delta(G^*) > \eps\Delta/6$, we return $\mathsf{FAIL}$.
If not, we properly color $G^*$ with $3\Delta(G^*)$ colors using Algorithm~\ref{alg:greedy}.
Let $\psi$ be the resulting coloring on $G^*$.
For each edge $e \in E(G^*)$, we let $\phi(e) = (1+\epsilon/2)\Delta + \psi(e)$.
Since $3\Delta(G^*) \leq \epsilon\Delta/2$, this defines a proper $(1+\epsilon)\Delta$-edge-coloring of $G$.

\vspace{5pt}
\begin{breakablealgorithm}\small
\caption{$(1+\eps)\Delta$-Edge-Coloring}\label{alg:main alg}
\begin{flushleft}
\textbf{Input}: An $n$-vertex graph $G$ having $m$ edges and maximum degree $\Delta$. \\
% \medskip
\textbf{Output}: A proper $(1 + \epsilon)\Delta$-edge-coloring of $G$.
\end{flushleft}

\begin{algorithmic}[1]
    \State $\phi(e) \gets \blank$ \textbf{for each} $e \in E$, \quad $U \gets E$
    \While{$U \neq \0$}
        \State Pick an edge $e \in U$ and a vertex $x \in e$ uniformly at random. \label{step:random_edge}
        \State $\phi \gets \hyperref[alg:color]{\mathsf{ColorOne}}(\phi, e, x, \kappa, \ell)$ \Comment{Algorithm~\ref{alg:color}}
        \State $U \gets U \setminus \set{e}$
    \EndWhile
    \medskip
    \State Let $G^*$ be the subgraph induced by the flagged edges. \label{step:define_G_star}
    \If{$\Delta(G^*) > \epsilon\,\Delta / 6$} 
        \State \Return $\mathsf{FAIL}$ \label{step:fail_stage2}
    \EndIf
    \State Run Algorithm~\ref{alg:greedy} to obtain a $3\Delta(G^*)$-edge-coloring $\psi$ of $G^*$ \label{step:sinnamon}
    \State $\phi(e) \gets \psi(e) + (1 + \epsilon/2)\Delta$ for each $e \in E(G^*)$ \label{step:combine}
    \State \Return $\phi$
\end{algorithmic}
\end{breakablealgorithm}
\vspace{5pt}

Note that at each iteration of the \textsf{while} loop above, the set $U$ contains the edges that are both uncolored and unflagged.
This will become clear once we describe Algorithm~\ref{alg:color}.
The rest of this section is split into two subsections.
In the first, we will formally describe the algorithm \hyperref[alg:color]{Color One}.
In the second, we will prove the correctness of Algorithm~\ref{alg:main alg}.

\subsection{Color One}\label{subsec: color one overview}

We will first provide an overview of our algorithm (formally stated as Algorithm~\ref{alg:color}).
The algorithm is split into subprocedures, some of which are taken verbatim from \cite{dhawan2024simple}.

The first subprocedure we will describe is \hyperref[alg:fan]{Make Fan}, which is stated formally as Algorithm~\ref{alg:fan} and is identical to \cite[Algorithm 3.1]{dhawan2024simple}. 
The algorithm takes as input a proper partial coloring $\phi$, an uncolored edge $e = xy$, a choice of a pivot vertex $x \in e$, and a subset of colors $C \subseteq [q]$.
Either the algorithm fails, or it outputs a tuple $(F, \alpha, j)$ such that:
\begin{itemize}
    \item $F = (x, y_0 = y, \ldots, y_{k-1})$ is a fan under the partial coloring $\phi$,
    \item $\alpha \in C$ is a color, and $j \leq k$ is an index such that $\alpha \in M(\phi, y_{k-1})\cap M(\phi, y_{j-1})$.
\end{itemize}
To construct $F$, we follow a series of iterations. At the start of each iteration, we have a fan $F = (x, y_0 = y, \ldots, y_s)$.
If $M(\phi, y_s) \cap C = \0$, the algorithm fails.
If not, we let $\eta \defeq \min M(\phi, y_s) \cap C$.
If $\eta \in M(\phi, x)$, then we return $(F, \eta, s+1)$.
If not, let $z$ be such that $\phi(xz) = \eta$. 
We now have two cases.
\begin{enumerate}[label=\ep{\textbf{Case \arabic*}},wide]
    \item $z\notin \set{y_0, \ldots, y_s}$. Then we update $F$ to $(x, y_0, \ldots, y_s, z)$ and continue.
    \item $z = y_j$ for some $0 \leq j \leq s$. 
    Note that $\phi(xy_0) = \blank$ and $\eta \in M(\phi, y_s)$, so we must have $1 \leq j \leq s - 1$.
    In this case, we return $(F, \eta, j)$.
\end{enumerate}

% \vspace{2pt}
\begin{breakablealgorithm}\small
\caption{Make Fan}\label{alg:fan}
\begin{flushleft}
\textbf{Input}: A proper partial $q$-edge-coloring $\phi$, an uncolored edge $e = xy$, a vertex $x \in e$, and list of colors $C \subseteq [q]$. \\
% \medskip
\textbf{Output}: $\mathsf{FAIL}$ or a fan $F = (x, y_0 = y, \ldots, y_{k-1})$, a color $\alpha \in C$, and an index $j$ such that $\alpha \in M(\phi, y_{k-1})\cap M(\phi, y_{j-1})$.
\end{flushleft}

\begin{algorithmic}[1]
    \State $F \gets (x,y)$, \quad $z \gets y$, \quad $k \gets 1$.
    \While{$\mathsf{true}$}
        \If{$C \cap M(\phi, z) = \0$} \label{step:check_missing}
            \State \Return $\mathsf{FAIL}$ \label{step:fail_fan}
        \EndIf
        \medskip
        \State $\eta \gets \min M(\phi, z) \cap C$ \label{step:assign_eta}
        \State $z \gets M(\phi, x)[\eta]$
        \If{$z = \blank$}
            \State \Return $(F, \eta, k)$ \label{step:happy_fan}
        \EndIf
        \medskip
        \If{$z \in F$}\label{step:check_fan}
            \State Let $j$ be such that $z = y_j$.
            \State \Return $(F, \eta, j)$
        \EndIf
        \medskip
        \State $y_k \gets z$
        \State $\mathsf{Append}(F, y_k)$, \quad $k \gets k + 1$
    \EndWhile \label{step:end}
\end{algorithmic}
\end{breakablealgorithm}
\vspace{5pt}

Note that Steps~\ref{step:check_missing} and \ref{step:assign_eta} can be implemented in $O(|C|)$ time and Step~\ref{step:check_fan} takes $O(k)$ time.
It is not hard to see that the number of iterations is at most $|C|$ and so $k \leq |C| + 1$.
It follows that Algorithm~\ref{alg:fan} runs in $O(|C|^2)$ time.

The next algorithm we will describe is the \hyperref[alg:viz]{Vizing Chain Algorithm}, which is stated formally as Algorithm~\ref{alg:viz} and is identical to \cite[Algorithm 3.2]{dhawan2024simple}. 
The algorithm takes as input a partial coloring $\phi$, an uncolored edge $e = xy$, a choice of a pivot vertex $x \in e$, a subset of colors $C \subseteq [q]$, and a parameter $\ell$ to be defined later.
Either the algorithm fails, or it outputs a Vizing chain $(F, P)$ and a color $\alpha$, which satisfy certain properties (see Lemma~\ref{lemma:viz}).
First, we run Algorithm~\ref{alg:fan} with input $(\phi, e, x, C)$.
If the algorithm fails, we return $\mathsf{FAIL}$.
If not, let $(F, \alpha, j)$ be the output of Algorithm~\ref{alg:fan} such that $F = (x, y_0, \ldots, y_{k-1})$.
If $j = k$, we return the Vizing chain $(F,(x))$ and the color $\alpha$.
If $M(\phi, x) \cap C = \0$, we return $\mathsf{FAIL}$.
If not, let $\beta \defeq \min M(\phi, x) \cap C$, and let $P = (x_0 = x, x_1, \ldots, x_s)$ be the $\alpha\beta$-path starting at $x$ such that $s \leq \ell$ (i.e., if the maximal $\alpha\beta$-path starting at $x$ has length $> \ell$, we only consider the first $\ell$ edges).
We return the Vizing chain $(F,P)$ and the color $\alpha$.

\vspace{5pt}
\begin{breakablealgorithm}\small
\caption{Vizing Chain Algorithm}\label{alg:viz}
\begin{flushleft}
\textbf{Input}: A proper partial coloring $\phi$, an uncolored edge $e = xy$, a choice of a pivot vertex $x \in e$, a subset of colors $C \subseteq [(1+\epsilon)\Delta]$, and a parameter $\ell$. \\
% \medskip
\textbf{Output}: $\mathsf{FAIL}$ or a Vizing chain $(F,P)$ and a color $\alpha$.
\end{flushleft}

\begin{algorithmic}[1]
    \State $\mathsf{Out} \gets \hyperref[alg:fan]{\mathsf{MakeFan}}(\phi, xy, x, C)$ \Comment{Algorithm~\ref{alg:fan}} \label{step:alg_fan}
    \If{$\mathsf{Out} = \mathsf{FAIL}$}
        \State \Return $\mathsf{FAIL}$
    \EndIf
    \medskip
    \State $(F, \alpha, j) \gets \mathsf{Out}$
    \If{$j = \length(F)$} 
        \State \Return $(F, (x)), \alpha$ \label{step:happy_fan_viz}
    \EndIf
    \medskip
    \If{$M(\phi, x) \cap C  = \0$}
        \State \Return $\mathsf{FAIL}$ \label{step:fail_x_alpha}
    \EndIf
    \medskip
    \State $\beta \gets \min M(\phi, x) \cap C$ \label{step:alpha}
    \State Let $P$ be the $\alpha\beta$-path of length at most $\ell$ starting at $x$ \label{step:path_def}
    \State \Return $(F, P), \alpha$.
\end{algorithmic}
\end{breakablealgorithm}
\vspace{5pt}

Note that we only compute a path of length at most $\ell$ at Step~\ref{step:path_def}.
Due to the structure of $M(\phi, \cdot)$, we conclude that we may implement Algorithm~\ref{alg:viz} in $O(|C|^2 + \ell)$ time.

The next subprocedure we will describe is an augmenting procedure for Vizing chains $(F, P)$.
The algorithm takes as input a partial coloring $\phi$, a Vizing chain $(F, P)$, and a color $\alpha$, and outputs a partial edge-coloring $\psi$ obtained by augmenting $\phi$ with $(F, P)$.
We say $\psi = \aug(\phi, (F, P))$.

\vspace{5pt}
\begin{breakablealgorithm}\small
\caption{Augmenting Algorithm}\label{alg:aug}
\begin{flushleft}
\textbf{Input}: A proper partial coloring $\phi$, a Vizing chain $(F, P)$, and a color $\alpha$. \\
% \medskip
\textbf{Output}: A partial coloring $\psi$ obtained by augmenting $\phi$ with $(F, P)$.
\end{flushleft}

\begin{algorithmic}[1]
    \State Let $F = (x, y_0, y_1, \ldots, y_{k-1})$ and $P = (x, x_1, \ldots, x_s)$.
    \If{$\exists j$ such that $y_j = x_1$}
        \State $F' \gets (x, y_0, \ldots, y_{j-1})$.
    \Else
        \State $F' \gets F$
    \EndIf
    \State Obtain $\psi'$ from $\phi$ by flipping the path $P$.
    \If{$x_s = \vend(F')$}
        \State Obtain $\psi$ from $\psi'$ by shifting the fan $F$ and coloring the edge $xy_{k-1}$ with $\alpha$.
    \Else
        \State Obtain $\psi$ from $\psi'$ by shifting the fan $F'$ and coloring the edge $xy_{j-1}$ with $\alpha$.
    \EndIf
    \State \Return $\psi$
\end{algorithmic}
\end{breakablealgorithm}
\vspace{5pt}

See Fig.~\ref{fig:cases_aug} for the possible situations in Algorithm~\ref{alg:aug}.
It is easy to see that the runtime is $O(\length(F) + \length(P))$.
Additionally, we will show that running Algorithm~\ref{alg:aug} with input $(\phi, (F, P), \alpha)$ outputs a proper partial coloring when $((F, P), \alpha)$ is obtained as the output of Algorithm~\ref{alg:viz} and $\length(P) < \ell$.

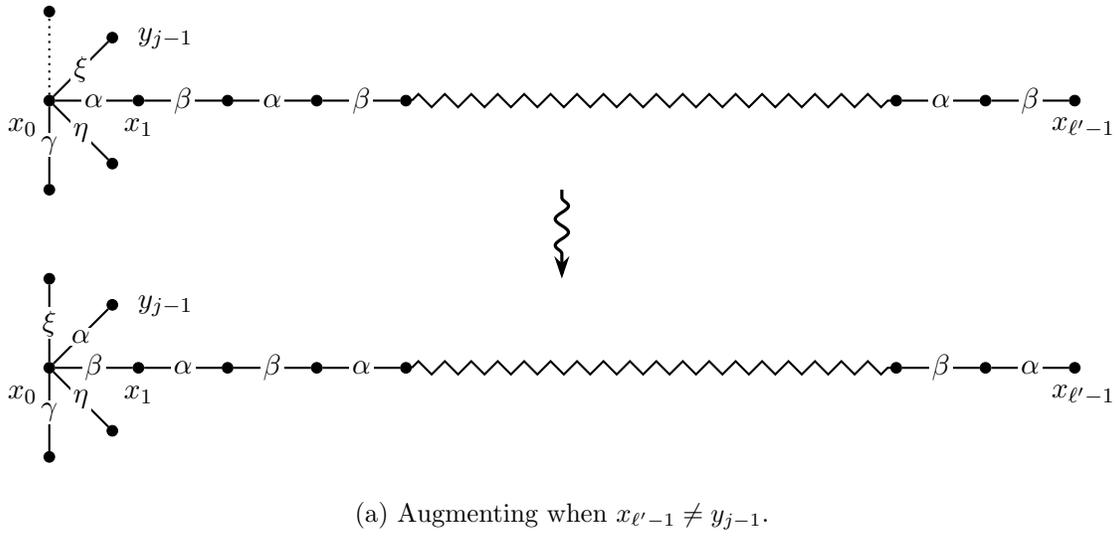
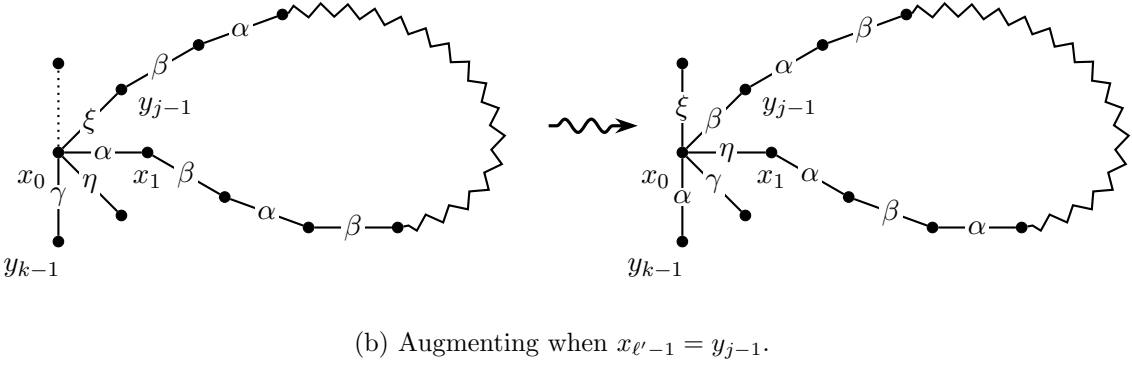
\begin{figure}[htb!]
        \begin{subfigure}[t]{\textwidth}
		\centering
		\begin{tikzpicture}[scale=1.185]
            \clip (-0.5, 4.5) rectangle (12, -1.3);
		\node[circle,fill=black,draw,inner sep=0pt,minimum size=4pt] (a) at (0,0) {};
        	\path (a) ++(0:1) node[circle,fill=black,draw,inner sep=0pt,minimum size=4pt] (b) {};
        	\path (b) ++(0:1) node[circle,fill=black,draw,inner sep=0pt,minimum size=4pt] (c) {};
                \path (c) ++(0:1) node[circle,fill=black,draw,inner sep=0pt,minimum size=4pt] (d) {};
                \path (d) ++(0:1) node[circle,fill=black,draw,inner sep=0pt,minimum size=4pt] (e) {};
                \path (e) ++(0:5.5) node[circle,fill=black,draw,inner sep=0pt,minimum size=4pt] (f) {};
                \path (f) ++(0:1) node[circle,fill=black,draw,inner sep=0pt,minimum size=4pt] (g) {};
                \path (g) ++(0:1) node[circle,fill=black,draw,inner sep=0pt,minimum size=4pt] (h) {};
        	
        	\draw[thick, decorate,decoration=zigzag] (e) to (f);

            \path (a) ++(45:1) node[circle,fill=black,draw,inner sep=0pt,minimum size=4pt] (i) {};
            \path (a) ++(90:1) node[circle,fill=black,draw,inner sep=0pt,minimum size=4pt] (j) {};
            \path (a) ++(-45:1) node[circle,fill=black,draw,inner sep=0pt,minimum size=4pt] (k) {};
            \path (a) ++(-90:1) node[circle,fill=black,draw,inner sep=0pt,minimum size=4pt] (l) {};

            \node at (-0.3, -0.3) {$x_0$};
            \node at (11.6, -0.3) {$x_{\ell'-1}$};
            \node at (1, -0.3) {$x_1$};
            \node at (1.3, 0.7) {$y_{j-1}$};
        	
        	\draw[thick] (a) to node[midway,inner sep=1pt,outer sep=1pt,minimum size=4pt,fill=white] {$\beta$} (b) to node[midway,inner sep=1pt,outer sep=1pt,minimum size=4pt,fill=white] {$\alpha$} (c) to node[midway,inner sep=1pt,outer sep=1pt,minimum size=4pt,fill=white] {$\beta$} (d) to node[midway,inner sep=1pt,outer sep=1pt,minimum size=4pt,fill=white] {$\alpha$} (e) (f) to node[midway,inner sep=1pt,outer sep=1pt,minimum size=4pt,fill=white] {$\beta$} (g) to node[midway,inner sep=1pt,outer sep=1pt,minimum size=4pt,fill=white] {$\alpha$} (h) (a) to node[midway,inner sep=1pt,outer sep=1pt,minimum size=4pt,fill=white] {$\alpha$} (i);

                \draw[thick] (a) to node[midway,inner sep=1pt,outer sep=1pt,minimum size=4pt,fill=white] {$\xi$} (j) (k) to node[midway,inner sep=1pt,outer sep=1pt,minimum size=4pt,fill=white] {$\eta$} (a) to node[midway,inner sep=1pt,outer sep=1pt,minimum size=4pt,fill=white] {$\gamma$} (l);
                % \draw[thick, dotted] (a) -- (j);

        \begin{scope}[yshift=3cm]
            \node[circle,fill=black,draw,inner sep=0pt,minimum size=4pt] (a) at (0,0) {};
        	\path (a) ++(0:1) node[circle,fill=black,draw,inner sep=0pt,minimum size=4pt] (b) {};
        	\path (b) ++(0:1) node[circle,fill=black,draw,inner sep=0pt,minimum size=4pt] (c) {};
                \path (c) ++(0:1) node[circle,fill=black,draw,inner sep=0pt,minimum size=4pt] (d) {};
                \path (d) ++(0:1) node[circle,fill=black,draw,inner sep=0pt,minimum size=4pt] (e) {};
                \path (e) ++(0:5.5) node[circle,fill=black,draw,inner sep=0pt,minimum size=4pt] (f) {};
                \path (f) ++(0:1) node[circle,fill=black,draw,inner sep=0pt,minimum size=4pt] (g) {};
                \path (g) ++(0:1) node[circle,fill=black,draw,inner sep=0pt,minimum size=4pt] (h) {};
        	
        	\draw[thick, decorate,decoration=zigzag] (e) to (f);

            \path (a) ++(45:1) node[circle,fill=black,draw,inner sep=0pt,minimum size=4pt] (i) {};
            \path (a) ++(90:1) node[circle,fill=black,draw,inner sep=0pt,minimum size=4pt] (j) {};
            \path (a) ++(-45:1) node[circle,fill=black,draw,inner sep=0pt,minimum size=4pt] (k) {};
            \path (a) ++(-90:1) node[circle,fill=black,draw,inner sep=0pt,minimum size=4pt] (l) {};

            \node at (-0.3, -0.3) {$x_0$};
            \node at (11.6, -0.3) {$x_{\ell'-1}$};
            \node at (1, -0.3) {$x_1$};
            \node at (1.3, 0.7) {$y_{j-1}$};

            \draw[thick] (a) to node[midway,inner sep=1pt,outer sep=1pt,minimum size=4pt,fill=white] {$\xi$} (i) (k) to node[midway,inner sep=1pt,outer sep=1pt,minimum size=4pt,fill=white] {$\eta$} (a) to node[midway,inner sep=1pt,outer sep=1pt,minimum size=4pt,fill=white] {$\gamma$} (l);
            \draw[thick, dotted] (a) -- (j);
        	
        	\draw[thick] (a) to node[midway,inner sep=1pt,outer sep=1pt,minimum size=4pt,fill=white] {$\alpha$} (b) to node[midway,inner sep=1pt,outer sep=1pt,minimum size=4pt,fill=white] {$\beta$} (c) to node[midway,inner sep=1pt,outer sep=1pt,minimum size=4pt,fill=white] {$\alpha$} (d) to node[midway,inner sep=1pt,outer sep=1pt,minimum size=4pt,fill=white] {$\beta$} (e) (f) to node[midway,inner sep=1pt,outer sep=1pt,minimum size=4pt,fill=white] {$\alpha$} (g) to node[midway,inner sep=1pt,outer sep=1pt,minimum size=4pt,fill=white] {$\beta$} (h);

        \end{scope}

        \begin{scope}[yshift=2.5cm]
            \draw[-{Stealth[length=3mm,width=2mm]},very thick,decoration = {snake,pre length=3pt,post length=7pt,},decorate] (5.75,-0.5) -- (5.75,-1.5);
        \end{scope}
        	
		\end{tikzpicture}
		\caption{Augmenting when $x_{\ell' - 1} \neq y_{j-1}$.}\label{subfig:case1}
	\end{subfigure}%
	\qquad%
	\begin{subfigure}[b]{\textwidth}
		\centering
		\begin{tikzpicture}[scale=1.185]
            \clip (-7.8, 2.3) rectangle (5.1, -1.8);
            \begin{scope}
		\node[circle,fill=black,draw,inner sep=0pt,minimum size=4pt] (a) at (0,0) {};
        	\path (a) ++(0:1) node[circle,fill=black,draw,inner sep=0pt,minimum size=4pt] (b) {};
        	\path (b) ++(-30:1) node[circle,fill=black,draw,inner sep=0pt,minimum size=4pt] (c) {};
                \path (c) ++(-20:1) node[circle,fill=black,draw,inner sep=0pt,minimum size=4pt] (d) {};
                \path (d) ++(0:1) node[circle,fill=black,draw,inner sep=0pt,minimum size=4pt] (e) {};

            \path (a) ++(45:1) node[circle,fill=black,draw,inner sep=0pt,minimum size=4pt] (i) {};
            \path (a) ++(90:1) node[circle,fill=black,draw,inner sep=0pt,minimum size=4pt] (j) {};
            \path (a) ++(-45:1) node[circle,fill=black,draw,inner sep=0pt,minimum size=4pt] (k) {};
            \path (a) ++(-90:1) node[circle,fill=black,draw,inner sep=0pt,minimum size=4pt] (l) {};

            \path (i) ++(30:1) node[circle,fill=black,draw,inner sep=0pt,minimum size=4pt] (f) {};
            \path (f) ++(20:1) node[circle,fill=black,draw,inner sep=0pt,minimum size=4pt] (g) {};

            \node at (-0.3, -0.3) {$x_0$};
            \node at (1.2, 0.5) {$y_{j-1}$};
            \node at (1, -0.3) {$x_1$};
            \node at (-0.3, -1.3) {$y_{k-1}$};

            \draw[thick] (a) to node[midway,inner sep=1pt,outer sep=1pt,minimum size=4pt,fill=white] {$\eta$} (b) (k) to node[midway,inner sep=1pt,outer sep=1pt,minimum size=4pt,fill=white] {$\gamma$} (a) to node[midway,inner sep=1pt,outer sep=1pt,minimum size=4pt,fill=white] {$\xi$} (j);
            
            \draw[thick, decorate, decoration=zigzag] (g) to[out=10,in=10, looseness=2] (e);
        	
        	\draw[thick] (a) to node[midway,inner sep=1pt,outer sep=1pt,minimum size=4pt,fill=white] {$\beta$} (i) (b) to node[midway,inner sep=1pt,outer sep=1pt,minimum size=4pt,fill=white] {$\alpha$} (c) to node[midway,inner sep=1pt,outer sep=1pt,minimum size=4pt,fill=white] {$\beta$} (d) to node[midway,inner sep=1pt,outer sep=1pt,minimum size=4pt,fill=white] {$\alpha$} (e) (g) to node[midway,inner sep=1pt,outer sep=1pt,minimum size=4pt,fill=white] {$\beta$} (f) to node[midway,inner sep=1pt,outer sep=1pt,minimum size=4pt,fill=white] {$\alpha$} (i) (a) to node[midway,inner sep=1pt,outer sep=1pt,minimum size=4pt,fill=white] {$\alpha$} (l);

         \end{scope}

        \begin{scope}[xshift=-7cm]
            \node[circle,fill=black,draw,inner sep=0pt,minimum size=4pt] (a) at (0,0) {};
        	\path (a) ++(0:1) node[circle,fill=black,draw,inner sep=0pt,minimum size=4pt] (b) {};
        	\path (b) ++(-30:1) node[circle,fill=black,draw,inner sep=0pt,minimum size=4pt] (c) {};
                \path (c) ++(-20:1) node[circle,fill=black,draw,inner sep=0pt,minimum size=4pt] (d) {};
                \path (d) ++(0:1) node[circle,fill=black,draw,inner sep=0pt,minimum size=4pt] (e) {};

            \path (a) ++(45:1) node[circle,fill=black,draw,inner sep=0pt,minimum size=4pt] (i) {};
            \path (a) ++(90:1) node[circle,fill=black,draw,inner sep=0pt,minimum size=4pt] (j) {};
            \path (a) ++(-45:1) node[circle,fill=black,draw,inner sep=0pt,minimum size=4pt] (k) {};
            \path (a) ++(-90:1) node[circle,fill=black,draw,inner sep=0pt,minimum size=4pt] (l) {};

            \path (i) ++(30:1) node[circle,fill=black,draw,inner sep=0pt,minimum size=4pt] (f) {};
            \path (f) ++(20:1) node[circle,fill=black,draw,inner sep=0pt,minimum size=4pt] (g) {};

            \node at (-0.3, -0.3) {$x_0$};
            \node at (1.2, 0.5) {$y_{j-1}$};
            \node at (1, -0.3) {$x_1$};
            \node at (-0.3, -1.3) {$y_{k-1}$};

            \draw[thick] (a) to node[midway,inner sep=1pt,outer sep=1pt,minimum size=4pt,fill=white] {$\xi$} (i) (k) to node[midway,inner sep=1pt,outer sep=1pt,minimum size=4pt,fill=white] {$\eta$} (a) to node[midway,inner sep=1pt,outer sep=1pt,minimum size=4pt,fill=white] {$\gamma$} (l);
            \draw[thick, dotted] (a) -- (j);
            
            \draw[thick, decorate, decoration=zigzag] (g) to[out=10,in=10, looseness=2] (e);
        	
        	\draw[thick] (a) to node[midway,inner sep=1pt,outer sep=1pt,minimum size=4pt,fill=white] {$\alpha$} (b) to node[midway,inner sep=1pt,outer sep=1pt,minimum size=4pt,fill=white] {$\beta$} (c) to node[midway,inner sep=1pt,outer sep=1pt,minimum size=4pt,fill=white] {$\alpha$} (d) to node[midway,inner sep=1pt,outer sep=1pt,minimum size=4pt,fill=white] {$\beta$} (e) (g) to node[midway,inner sep=1pt,outer sep=1pt,minimum size=4pt,fill=white] {$\alpha$} (f) to node[midway,inner sep=1pt,outer sep=1pt,minimum size=4pt,fill=white] {$\beta$} (i);
        \end{scope}

        \draw[-{Stealth[length=3mm,width=2mm]},very thick,decoration = {snake,pre length=3pt,post length=7pt,},decorate] (-1.5,0.3) -- (-0.5,0.3);
		\end{tikzpicture}
		\caption{Augmenting when $x_{\ell' - 1} = y_{j-1}$.}\label{subfig:case2}
	\end{subfigure}
	\caption{Possible situations of Algorithm~\ref{alg:aug}.}\label{fig:cases_aug}
\end{figure}

We are now ready to describe our algorithm.
The formal description is in Algorithm~\ref{alg:color}, but we first provide an informal overview.
The algorithm takes as input a proper partial coloring $\phi$, an uncolored edge $e = xy$, a choice of a pivot vertex $x \in e$, and two parameters $\kappa$ and $\ell$, and outputs a proper partial edge-coloring $\psi$ such that 
\begin{itemize}
    \item either $\dom(\psi) = \dom(\phi) \cup \set{e}$ and $\flg(\psi) = \flg(\phi)$, or
    \item there is an edge $f \in \dom(\phi) \cup \set{e}$ such that $\dom(\psi) = (\dom(\phi) \cup \set{e}) \setminus \set{f}$ and $\flg(\psi) = \flg(\phi) \cup \set{f}$.
\end{itemize}

As mentioned in \S\ref{subsection: proof overview}, we will do so by repeatedly ``shifting'' the uncolored edge $e$ by invoking Algorithm~\ref{alg:viz} to construct Vizing chains using disjoint palettes until one of the following holds:
\begin{enumerate}[label=(T\arabic*), wide]
    \item\label{item: fail vizing} Algorithm~\ref{alg:viz} fails,
    \item\label{item: success vizing} the path in the output of Algorithm~\ref{alg:viz} has length $< \ell$, or
    \item\label{item: empty palette vizing} we max out on iterations.
\end{enumerate}
Let us now describe the shifting procedure.

Let $\phi_t$ denote the coloring before the $t$-th call to Algorithm~\ref{alg:viz}, let $e_t = x_ty_t$ denote the uncolored and unflagged edge, and let $x_t$ denote the choice of the pivot vertex.
First, we define $C_t \subseteq [q]$ by sampling with replacement $\kappa$ times from $[(1 + \epsilon/2)\Delta] \setminus \left(\cup_{j < t}C_j\right)$.
We will run Algorithm~\ref{alg:viz} with input $(\phi_t, e_t, x_t, C_t, \ell)$.
If the output is $\mathsf{FAIL}$, we flag the edge $e_t$ and return $\phi_t$.
If not, let $((F, P), \alpha)$ be the output such that $F = (x_t, y_0, \ldots, y_{k-1})$ and $P = (x_0 = x_t, x_1, \ldots, x_s)$.
If $s < \ell$, then augment $\phi_t$ with $(F, P)$ by applying Algorithm~\ref{alg:aug}, and return the resulting coloring.
If not, let $\ell' \in [\ell]$ be chosen uniformly at random and let $F' = (x, y_0, \ldots, y_{j-1})$, where $j$ is such that $y_j = x_1$.
Uncolor the edge $x_{\ell' - 1}x_{\ell'}$ to obtain the coloring $\psi_t$, and let $P' = (x_0 = x, \ldots, x_{\ell' - 1})$.
In this case, we let $e_{t+1} = x_{\ell' - 1}x_{\ell'}$, $x_{t+1} = x_{\ell' - 1}$, and let $\phi_{t+1}$ be defined by flipping $P'$, then shifting $F'$, and finally coloring $x_ty_{j-1}$ with $\alpha$ (we will show this operation is valid in \S\ref{subsec: proof of correctness}).

\vspace{5pt}
\begin{breakablealgorithm}\small
\caption{Color One}\label{alg:color}
\begin{flushleft}
\textbf{Input}: A proper partial coloring $\phi$, an uncolored edge $e = xy$, a choice of a pivot vertex $x \in e$, and two parameters $\kappa$ and $\ell$. \\
% \medskip
\textbf{Output}: A proper partial edge-coloring $\psi$ such that either $\dom(\psi) = \dom(\phi) \cup \set{e}$ and $\flg(\psi) = \flg(\phi)$, or there is an edge $f \in \dom(\phi) \cup \set{e}$ such that $\dom(\psi) = (\dom(\phi) \cup \set{e}) \setminus \set{f}$ and $\flg(\psi) = \flg(\phi) \cup \set{f}$.
\end{flushleft}

\begin{algorithmic}[1]
    \State $Q \gets [(1+\eps/2)\Delta]$
    \For{$t = 1, \ldots, 100\log \Delta$}
        \State Define $C$ by sampling with replacement $\kappa$ times from $Q$. \label{step:random_sample}
        \State $Q \gets Q \setminus C$
        \State $\mathsf{Out} \gets \hyperref[alg:viz]{\mathsf{VizingChain}}(\phi, e, x, C, \ell)$ 
        \medskip\Comment{Algorithm~\ref{alg:viz}} \label{step:call_to_viz}
        \If{$\mathsf{Out} = \mathsf{FAIL}$}
            \State Flag the edge $e$. \label{step: flag at fail}
            \State \Return $\phi$ \label{step:fail_flag}
        \EndIf
        \medskip
        \State $((F, P), \alpha) \gets \mathsf{Out}$
        \State Let $F = (x, y_0, \ldots, y_{k-1})$ and let $P = (x_0 = x, x_1, \ldots, x_s)$.
        \If{$s < \ell$}
            \State $\phi \gets \hyperref[alg:aug]{\aug}(\phi, (F, P), \alpha)$. \Comment{Algorithm~\ref{alg:aug}} \label{step: aug}
            \State \Return $\phi$
        \Else \label{step:start}
            \State Let $\ell' \in [\ell]$ be an integer chosen uniformly at random. \label{step:random_choice}
            \State Uncolor the edge $x_{\ell'-1}x_{\ell'}$. \label{step:uncolor_edge}
            \State $P' \gets (x_0, \ldots, x_{\ell'-1})$.
            \State Let $j$ be such that $y_j = x_1$ and let $F' = (x, y_0, \ldots, y_{j-1})$.
            \State Update $\phi$ by flipping $P'$, then shifting $F'$, and finally coloring $xy_{j-1}$ with $\alpha$. \label{step:fake_aug}
            \State $e \gets x_{\ell'-1}x_{\ell'}$, \quad $x \gets x_{\ell'-1}$
        \EndIf \label{step:end_for}
    \EndFor
    \medskip
    \State Flag the edge $e$. \label{step: flag at end}
    \State \Return $\phi$
\end{algorithmic}
\end{breakablealgorithm}
\vspace{5pt}

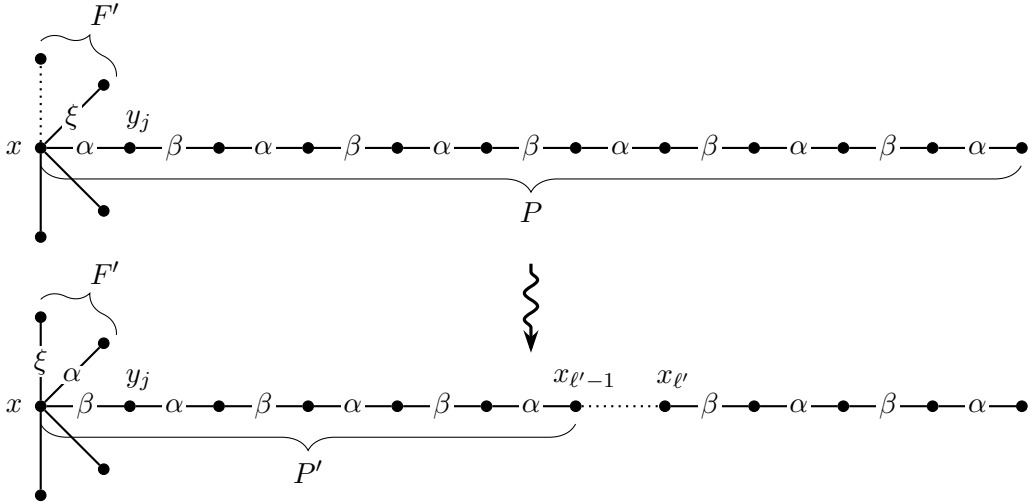
\begin{figure}[htb!]
	\centering
	% \begin{subfigure}[t]{\textwidth}
		% \centering
		\begin{tikzpicture}[scale=1.185]
            \node[circle,fill=black,draw,inner sep=0pt,minimum size=4pt] (a) at (0,0) {};
            \path (a) ++(45:1) node[circle,fill=black,draw,inner sep=0pt,minimum size=4pt] (w) {};
            \path (a) ++(-45:1) node[circle,fill=black,draw,inner sep=0pt,minimum size=4pt] (x) {};
            \path (a) ++(90:1) node[circle,fill=black,draw,inner sep=0pt,minimum size=4pt] (y) {};
            \path (a) ++(-90:1) node[circle,fill=black,draw,inner sep=0pt,minimum size=4pt] (z) {};
            
        	\path (a) ++(0:1) node[circle,fill=black,draw,inner sep=0pt,minimum size=4pt] (b) {};
        	\path (b) ++(0:1) node[circle,fill=black,draw,inner sep=0pt,minimum size=4pt] (c) {};
                \path (c) ++(0:1) node[circle,fill=black,draw,inner sep=0pt,minimum size=4pt] (d) {};
                \path (d) ++(0:1) node[circle,fill=black,draw,inner sep=0pt,minimum size=4pt] (e) {};
                \path (e) ++(0:1) node[circle,fill=black,draw,inner sep=0pt,minimum size=4pt] (l) {};
                \path (l) ++(0:1) node[circle,fill=black,draw,inner sep=0pt,minimum size=4pt] (m) {};
                \path (m) ++(0:1) node[circle,fill=black,draw,inner sep=0pt,minimum size=4pt] (n) {};
                \path (n) ++(0:1) node[circle,fill=black,draw,inner sep=0pt,minimum size=4pt] (o) {};
                \path (o) ++(0:1) node[circle,fill=black,draw,inner sep=0pt,minimum size=4pt] (f) {};
                \path (f) ++(0:1) node[circle,fill=black,draw,inner sep=0pt,minimum size=4pt] (g) {};
                \path (g) ++(0:1) node[circle,fill=black,draw,inner sep=0pt,minimum size=4pt] (h) {};
         
        	\draw[thick] (a) to node[midway,inner sep=1pt,outer sep=1pt,minimum size=4pt,fill=white] {$\xi$} (y) (a) to node[midway,inner sep=1pt,outer sep=1pt,minimum size=4pt,fill=white] {$\alpha$} (w) (a) -- (x) (a) -- (z)(a) to node[midway,inner sep=1pt,outer sep=1pt,minimum size=4pt,fill=white] {$\beta$} (b) to node[midway,inner sep=1pt,outer sep=1pt,minimum size=4pt,fill=white] {$\alpha$} (c) to node[midway,inner sep=1pt,outer sep=1pt,minimum size=4pt,fill=white] {$\beta$} (d) to node[midway,inner sep=1pt,outer sep=1pt,minimum size=4pt,fill=white] {$\alpha$} (e) to node[midway,inner sep=1pt,outer sep=1pt,minimum size=4pt,fill=white] {$\beta$} (l) to node[midway,inner sep=1pt,outer sep=1pt,minimum size=4pt,fill=white] {$\alpha$} (m)  (n) to node[midway,inner sep=1pt,outer sep=1pt,minimum size=4pt,fill=white] {$\beta$} (o) to node[midway,inner sep=1pt,outer sep=1pt,minimum size=4pt,fill=white] {$\alpha$} (f) to node[midway,inner sep=1pt,outer sep=1pt,minimum size=4pt,fill=white] {$\beta$} (g) to node[midway,inner sep=1pt,outer sep=1pt,minimum size=4pt,fill=white] {$\alpha$} (h);

                \draw[thick, dotted] (m) -- (n);

            \draw[decoration={brace,amplitude=10pt,mirror}, decorate] (0, -0.2) -- node [midway,below,xshift=0pt,yshift=-10pt] {$P'$} (6,-0.2);

            \draw[decoration={brace,amplitude=10pt}, decorate] (0, 1.2) -- node [midway,below,xshift=10pt,yshift=25pt] {$F'$} (0.85,0.8);

            \node at (6.1, 0.3) {$x_{\ell' - 1}$};
            \node at (7.1, 0.3) {$x_{\ell'}$};
            \node at (-0.3, 0) {$x$};
            \node at (1.1, 0.3) {$y_j$};

        \begin{scope}[yshift=2.9cm]
            \node[circle,fill=black,draw,inner sep=0pt,minimum size=4pt] (a) at (0,0) {};
            \path (a) ++(45:1) node[circle,fill=black,draw,inner sep=0pt,minimum size=4pt] (w) {};
            \path (a) ++(-45:1) node[circle,fill=black,draw,inner sep=0pt,minimum size=4pt] (x) {};
            \path (a) ++(90:1) node[circle,fill=black,draw,inner sep=0pt,minimum size=4pt] (y) {};
            \path (a) ++(-90:1) node[circle,fill=black,draw,inner sep=0pt,minimum size=4pt] (z) {};
            
        	\path (a) ++(0:1) node[circle,fill=black,draw,inner sep=0pt,minimum size=4pt] (b) {};
        	\path (b) ++(0:1) node[circle,fill=black,draw,inner sep=0pt,minimum size=4pt] (c) {};
                \path (c) ++(0:1) node[circle,fill=black,draw,inner sep=0pt,minimum size=4pt] (d) {};
                \path (d) ++(0:1) node[circle,fill=black,draw,inner sep=0pt,minimum size=4pt] (e) {};
                \path (e) ++(0:1) node[circle,fill=black,draw,inner sep=0pt,minimum size=4pt] (l) {};
                \path (l) ++(0:1) node[circle,fill=black,draw,inner sep=0pt,minimum size=4pt] (m) {};
                \path (m) ++(0:1) node[circle,fill=black,draw,inner sep=0pt,minimum size=4pt] (n) {};
                \path (n) ++(0:1) node[circle,fill=black,draw,inner sep=0pt,minimum size=4pt] (o) {};
                \path (o) ++(0:1) node[circle,fill=black,draw,inner sep=0pt,minimum size=4pt] (f) {};
                \path (f) ++(0:1) node[circle,fill=black,draw,inner sep=0pt,minimum size=4pt] (g) {};
                \path (g) ++(0:1) node[circle,fill=black,draw,inner sep=0pt,minimum size=4pt] (h) {};

            \draw[thick, dotted] (a) -- (y);
        	\draw[thick] (a) to node[midway,inner sep=1pt,outer sep=1pt,minimum size=4pt,fill=white] {$\xi$} (w) (a) -- (x) (a) -- (z) (a) to node[midway,inner sep=1pt,outer sep=1pt,minimum size=4pt,fill=white] {$\alpha$} (b) to node[midway,inner sep=1pt,outer sep=1pt,minimum size=4pt,fill=white] {$\beta$} (c) to node[midway,inner sep=1pt,outer sep=1pt,minimum size=4pt,fill=white] {$\alpha$} (d) to node[midway,inner sep=1pt,outer sep=1pt,minimum size=4pt,fill=white] {$\beta$} (e) to node[midway,inner sep=1pt,outer sep=1pt,minimum size=4pt,fill=white] {$\alpha$} (l) to node[midway,inner sep=1pt,outer sep=1pt,minimum size=4pt,fill=white] {$\beta$} (m) to node[midway,inner sep=1pt,outer sep=1pt,minimum size=4pt,fill=white] {$\alpha$} (n) to node[midway,inner sep=1pt,outer sep=1pt,minimum size=4pt,fill=white] {$\beta$} (o) to node[midway,inner sep=1pt,outer sep=1pt,minimum size=4pt,fill=white] {$\alpha$} (f) to node[midway,inner sep=1pt,outer sep=1pt,minimum size=4pt,fill=white] {$\beta$} (g) to node[midway,inner sep=1pt,outer sep=1pt,minimum size=4pt,fill=white] {$\alpha$} (h);

            \draw[decoration={brace,amplitude=10pt,mirror}, decorate] (0, -0.2) -- node [midway,below,xshift=0pt,yshift=-10pt] {$P$} (11,-0.2);
            \draw[decoration={brace,amplitude=10pt}, decorate] (0, 1.2) -- node [midway,below,xshift=10pt,yshift=25pt] {$F'$} (0.85,0.8);
            \node at (-0.3, 0) {$x$};
            \node at (1.1, 0.3) {$y_j$};

        \end{scope}

        \begin{scope}[yshift=2.1cm]
            \draw[-{Stealth[length=3mm,width=2mm]},very thick,decoration = {snake,pre length=3pt,post length=7pt,},decorate] (5.5,-0.5) -- (5.5,-1.5);
        \end{scope}
        	
        \end{tikzpicture}
		\caption{Shifting an uncolored edge from Steps~\ref{step:start}--\ref{step:end_for} of Algorithm~\ref{alg:color}.}\label{fig:flag}
\end{figure}

Note the following for each iteration:
\begin{align*}
    |Q| \,\geq\, (1 + \eps/2)\Delta - 100\,\kappa\,\log \Delta.
\end{align*}
For our choice of $\kappa$, we will show the above is always at least $(1 + \eps/100)\Delta$.
As $\length(P') \leq \ell$ and by the earlier observations made on the runtime of Algorithms~\ref{alg:viz} and \ref{alg:aug}, it follows that the runtime of \hyperref[alg:color]{Color One} is
\begin{align}\label{eq: runtime}
    O\left(m\,(\kappa^2 + \ell)\,\log \Delta\right).
\end{align}

\subsection{Proof of Correctness}\label{subsec: proof of correctness}

We begin with a few lemmas from \cite{dhawan2024simple} concerning Algorithms~\ref{alg:fan} and \ref{alg:viz}.
The following describes certain properties of the output of Algorithm~\ref{alg:fan}:

\begin{lemma}[{\cite[Lemma 3.1]{dhawan2024simple}}]\label{lemma:fan}
    Consider running Algorithm~\ref{alg:fan} on input $(\phi, xy, x, C)$.
    Suppose the algorithm does not fail and outputs $(F, \alpha, j)$ where $F = (x, y_0 = y, y_1, \ldots, y_{k-1})$.
    Let $\beta \in M(\phi, x)$ be arbitrary and let $P$ be the maximal $\alpha\beta$-path under $\phi$ starting at $x$.
    One of the following must hold:
    \begin{enumerate}
        \item either $j = k$ and the edge $xy_{k-1}$ is $\psi$-happy, where $\psi$ is the coloring obtained from $\phi$ by shifting $F$, or
        \item the edge $xy_{j-1}$ is $\psi$-happy, where $\psi$ is the coloring obtained from $\phi$ by flipping $P$ and then shifting the fan $F' = (x, y_0 = y, y_1, \ldots, y_{j-1})$, or
        \item the edge $xy_{k-1}$ is $\psi$-happy, where $\psi$ is the coloring obtained from $\phi$ by flipping $P$ and then shifting $F$.
    \end{enumerate}
    Moreover, the happy edge can be colored with the color $\alpha$.
\end{lemma}

The next lemma describes certain properties of the output of Algorithm~\ref{alg:viz}:

\begin{lemma}[{\cite[Lemma 3.2]{dhawan2024simple}}]\label{lemma:viz}
    Consider running Algorithm~\ref{alg:viz} on input $(\phi, xy, x, C, \ell)$.
    Suppose the algorithm does not fail and outputs $((F, P), \alpha)$ such that $F = (x, y_0=y, \ldots, y_{k-1})$ and $P = (x_0 = x, \ldots, x_s)$ is an $\alpha\beta$-path.
    One of the following must hold for $j$ such that $y_j = x_1$ (if $s\geq 1$):
    \begin{enumerate}
        \item either $P = (x)$ and the edge $xy_{k-1}$ is $\psi$-happy, where $\psi$ is the coloring obtained from $\phi$ by shifting $F$, or
        \item $\length(P) < \ell$ and the edge $xy_{j-1}$ is $\psi$-happy, where $\psi$ is the coloring obtained from $\phi$ by flipping $P$ and then shifting $F' = (x, y_0 = y, y_1, \ldots, y_{j-1})$, or
        \item $\length(P) < \ell$ and the edge $xy_{k-1}$ is $\psi$-happy, where $\psi$ is the coloring obtained from $\phi$ by flipping $P$ and then shifting $F$, or
        \item $\length(P) = \ell$.
    \end{enumerate}
    Moreover, the happy edge may be colored with the color $\alpha$.
\end{lemma}

We obtain the following lemma as a corollary to Lemma~\ref{lemma:viz}:

\begin{lemma}\label{lemma:aug} 
    Let $((F, P), \alpha)$ be the output after running Algorithm~\ref{alg:viz} on input $(\phi, xy, x, C, \ell)$.
    If $\length(P) < \ell$, then running Algorithm~\ref{alg:aug} on input $(\phi, (F, P), \alpha)$ produces a proper partial edge-coloring $\psi$ in time $O(|C|^2 + \ell)$.
\end{lemma}

In particular, the above implies that Step~\ref{step: aug} results in a proper edge-coloring.
Let us now show the same for Step~\ref{step:fake_aug}.

\begin{lemma}\label{lemma: fake_aug}
    Step~\ref{step:fake_aug} results in a proper edge-coloring. 
\end{lemma}

\begin{proof}
    Let $\phi$ be the proper edge-coloring, $e$ be the uncolored edge, and $x \in e$ be the pivot vertex at the start of the iteration.
    Additionally, let $C$ be the set of colors sampled at Step~\ref{step:random_sample}.
    Suppose $((F, P), \alpha)$ is the output at Step~\ref{step:call_to_viz} such that $\length(P) = \ell$.
    Let $\ell'$ be the value sampled at Step~\ref{step:random_choice} and let $F'$ and $P'$ be as defined.
    Finally, let $\psi$ be the coloring obtained by uncoloring $x_{\ell' - 1}x_{\ell'}$ at Step~\ref{step:uncolor_edge}.

    If $\ell \leq 1$, we have $\vend(P') \neq y_{j-1}$, trivially.
    For $\ell \geq 2$, we note that $\vend(P') \neq y_{j-1}$ as $\alpha \in M(\phi, y_{j-1})$ and not in $M(\phi, \vend(P'))$.
    The claim now follows by an identical argument to the proof of \cite[Lemma~3.1(2)]{dhawan2024simple}.
\end{proof}

The following result follows by construction and by the preceding lemmas in this section:

\begin{lemma}\label{lemma: correct}
    Let $\psi$ be the output after running Algorithm~\ref{alg:color} with input $(\phi, e, x, \kappa, \ell)$.
    Then, 
    \begin{itemize}
        \item either $\dom(\psi) = \dom(\phi) \cup \set{e}$ and $\flg(\psi) = \flg(\phi)$, or
        \item there is an edge $f \in \dom(\phi) \cup \set{e}$ such that $\dom(\psi) = (\dom(\phi) \cup \set{e}) \setminus \set{f}$ and $\flg(\psi) = \flg(\phi) \cup \set{f}$.
    \end{itemize}
\end{lemma}

Note that the above implies that $\dom(\phi) \cup \flg(\phi) \subseteq \dom(\psi) \cup \flg(\psi)$.
Furthermore, there is precisely one additional edge in the latter set (the edge $e$).
In particular, this justifies the way we control the set $U$ in Algorithm~\ref{alg:main alg}.

%% file: proof.tex
\section{Proof of Theorem~\ref{theo:main_theo}}\label{section: proof}

Let us run Algorithm~\ref{alg:main alg} with $\kappa = \Theta(\log \Delta/\eps)$ and $\ell = \Theta(\kappa^2)$, where the implicit constants are assumed to be sufficiently large.
By these choices of parameters and by \eqref{eq: runtime}, it follows that the \textsf{while} loop takes $O(m\log^3\Delta/\eps^2)$ time.
As we may implement Step~\ref{step:fail_stage2} in $O(m)$ time, the above observation along with Proposition~\ref{prop:greedy} imply Algorithm~\ref{alg:main alg} has running time $O\left(m\log^3\Delta/\eps^2\right)$.
It remains to bound the probability of failure.
The main result of this section is the following:

\begin{proposition}\label{prop:failure}
    Consider running Algorithm~\ref{alg:main alg} with input $G$ and parameters $\kappa$ and $\ell$ as defined at the beginning of this section.
    Then, 
    \[\Pr[\text{{\upshape{we reach Step~\ref{step:fail_stage2}}}}] \,\leq\, \dfrac{1}{\poly(n)}.\]
\end{proposition}

In order to prove the above, we make the following definitions for each of the $m$ iterations of the \textsf{while} loop in Algorithm~\ref{alg:main alg}:
\begin{align*}
    \phi_i &\defeq \text{ the coloring at the start of the $i$-th iteration}, \\
    f_i &\defeq \text{ the edge flagged during the corresponding call to \hyperref[alg:color]{Color One}}.
\end{align*}
Note that $f_i = \blank$ if we reach Step~\ref{step: aug} during the $i$-th call to \hyperref[alg:color]{Color One}.
Let $G^*$ be the subgraph at Step~\ref{step:define_G_star} of Algorithm~\ref{alg:main alg}.

We will define the following random variables for each vertex $v \in V(G)$:
\[d_i(v) \defeq \bbone{v \in f_i}, \quad d(v) \defeq \deg_{G^*}(v) = \sum_{i = 1}^md_i(v).\]
In the following lemma, we shall bound $\Pr[d_i(v) = 1\mid \phi_i]$.

\begin{lemma}\label{lemma: degree of v in G star}
    $\Pr[d_i(v) = 1\mid \phi_i] \leq \dfrac{1}{(m - i + 1)\poly(\Delta)}$.
\end{lemma}

\begin{proof}
    Note that $d_i(v) = 1$ if $v \in f_i$.
    Furthermore, $f_i \neq \blank$ if we either reach Step~\ref{step: flag at fail} during some iteration of the \textsf{while} loop or Step~\ref{step: flag at end}.
    Let us first consider the latter.

    Let $T = \Theta(\log \Delta)$, where the implicit constant is assumed to be sufficiently large.
    To assist with the proof, we make a few definitions for each of the $1 \leq t \leq T$ iterations of the \textsf{while} loop during the $i$-th call to Algorithm~\ref{alg:color}:
    \begin{align*}
        \psi_t &\defeq \text{the coloring at the end of the $t$-th iteration}, \\
        e_t &\defeq \text{the uncolored edge at the end of the $t$-th iteration, i.e., after shifting}, \\
        x_t &\defeq \text{the pivot vertex for the $(t+1)$-th iteration}, \\
        C_t &\defeq \text{the palette sampled at Step~\ref{step:random_sample} during the $t$-th iteration}.
    \end{align*}
    Additionally, we let $\psi_0 \defeq \phi_i$ and let $e_0$ and $x_0$ denote the random edge and vertex sampled at Step~\ref{step:random_edge} of Algorithm~\ref{alg:main alg}.
    The following claim will assist with our proof:

    \begin{claim}\label{claim: related}
        For every $1 \leq t \leq T$, there exist $\alpha,\, \beta \in C_t$ such that 
        \begin{itemize}
            \item $x_{t-1}$ and $x_t$ lie on an $\alpha\beta$-path under $\phi_i$, and
            \item $x_{t-1}$ is an endpoint of this path.
        \end{itemize}
    \end{claim}
    
    \begin{claimproof}
        By construction, the claim holds for $\phi_i$ replaced by $\psi_{t-1}$, where $\alpha$ and $\beta$ are the colors on the path in the Vizing chain constructed at Step~\ref{step:call_to_viz} during the $t$-th iteration of the \textsf{while} loop.
        Let $Q = (v_0 = x_{t-1}, v_1, \ldots, v_k = x_t)$ be the (not necessarily maximal) $\alpha\beta$-path under $\psi_{t-1}$ connecting the vertices in question.
    
        If $Q$ does not form an $\alpha\beta$-path under $\phi_i$, there is some edge $f = v_jv_{j+1}$ such that $\phi_i(f) \neq \psi_{t-1}(f)$.
        In particular, for some $s < t$, we have $\psi_{s-1}(f) \neq \psi_s(f)$ and $\psi_s(f) \in \set{\alpha, \beta}$.
        It follows that $f$ lies on the Vizing chain constructed during the $s$-th iteration.
        As this Vizing chain was constructed using the palette $C_s$, we conclude that $C_s \cap \set{\alpha, \beta} \neq \0$.
        This is a contradiction as the palettes $C_1, \ldots, C_T$ are disjoint by construction.
        Therefore, $Q$ must form an $\alpha\beta$-path under $\phi_i$.
    
        Let $P$ be the maximal $\alpha\beta$-path containing $Q$ under $\phi_i$.
        Suppose $x_{t-1}$ is not an endpoint of $P$.
        Then, there is some edge $x_{t-1}v$ on $P$ such that $\phi_i(x_{t-1}v) \neq \psi_{t-1}(x_{t-1}v)$.
        We arrive at a contradiction in an identical fashion as in the previous paragraph.
    \end{claimproof}
    
    Note that if $v \in f_i$, there exists an edge $e$ and vertices $v_0, \ldots, v_T$ such that
    \begin{itemize}
        \item $e_0 = e$, 
        \item $x_t = v_t$ for each $0 \leq t \leq T$, and
        \item $v_T \in N[v]$.
    \end{itemize}
    Furthermore, as the random choices made at Step~\ref{step:random_edge} of Algorithm~\ref{alg:main alg} and Step~\ref{step:random_choice} of Algorithm~\ref{alg:color} are independent, we have
    \[\Pr\left[e_0 = e, x_0 = v_0, \ldots, x_t = v_t\mid \phi_i, C_1, \ldots, C_T\right] \,\leq\, \frac{1}{2(m-i+1)\,\ell^t},\]
    where we use the fact that $|U| = m - i + 1$ during the $i$-th iteration of the \textsf{while} loop of Algorithm~\ref{alg:main alg} (this can be concluded as a result of Lemma~\ref{lemma: correct}).
    
    Let us show that there are at most $2\,\Delta^2\,\kappa^{2T}$ such tuples $(e, v_0, \ldots, v_T)$.
    First, we note that by Claim~\ref{claim: related}, the vertex $v_{t-1}$ is the endpoint of an $\alpha_t\beta_t$-path passing through $v_t$ under $\phi_i$ such that $\alpha_t,\,\beta_t \in C_t$.
    In particular, given $v_t$, there are at most $2\binom{|C_t|}{2}$ choices for $v_{t-1}$.
    As $2\binom{|C_t|}{2} \leq \kappa^2$, given $v_T$, there are at most $\kappa^{2T}$ possible choices for $(v_0, \ldots, v_{T-1})$.
    Next, there are at most $\Delta + 1 \leq 2\Delta$ choices for $v_T$ as $v_T \in N[v]$.
    Finally, given $v_0$, there are at most $\Delta$ choices for $e$.
    
    Recall our choices of the parameters $\kappa$, $\ell$, and $T$: $\kappa = \Theta(\log \Delta/\eps)$, $\ell = \Theta(\kappa^2)$, and $T = 100\log \Delta$.
    It follows that
    \begin{align*}
        \Pr[d_i(v) = 1 \text{ and we reach Step~\ref{step: flag at end}} \mid \phi_i, C_1, \ldots, C_T] &\leq\, \frac{\Delta^2}{(m - i + 1)}\left(\frac{\kappa^2}{\ell}\right)^T \\
        &\leq\, \frac{1}{(m-i+1)\poly(\Delta)},
    \end{align*}
    for sufficiently large hidden constants in the $\Theta(\cdot)$ notation.
    % where the last step follows for our choices of $\kappa$, $\ell$, and $T$.
    As the above is independent of $C_1, \ldots, C_T$, we conclude that
    \begin{align}\label{eq: fail end}
        \Pr[d_i(v) = 1 \text{ and we reach Step~\ref{step: flag at end}} \mid \phi_i] \leq \frac{1}{(m-i+1)\poly(\Delta)}.
    \end{align}

    Now, let us consider the case where we reach Step~\ref{step: flag at fail}.
    If we have $v \in f_i$ and we reach Step~\ref{step: flag at fail}, the following must hold for some iteration $1 \leq t \leq T$: there exists $u \in N[v]$ such that $x_{t-1} = u$ and there exists $w \in N[u]$ such that $M(\psi_{t - 1}, w) \cap C_t = \0$.
    Therefore, we have the following:
    \begin{align*}
        &~\Pr[d_i(v) = 1 \text{ and we reach Step~\ref{step: flag at fail}} \mid \phi_i] \\
        &\leq \sum_{t = 1}^T\sum_{u \in N[v]}\Pr[x_{t-1} = u,\,\exists w \in N[u] \text{ s.t. } M(\psi_{t - 1}, w) \cap C_t = \0 \mid \phi_i] \\
        &\leq \sum_{t = 1}^T\sum_{u \in N[v]}\Pr[\exists w \in N[u] \text{ s.t. } M(\psi_{t - 1}, w) \cap C_t = \0 \mid \phi_i,\,x_{t-1} = u]\,\Pr[x_{t-1} = u \mid \phi_i].
    \end{align*}
    Following an identical process as earlier, we may conclude that
    \[\Pr[x_{t-1} = u \mid \phi_i, C_1, \ldots, C_{t-1}] \leq \frac{\Delta}{2(m - i + 1)}\left(\frac{\kappa^2}{\ell}\right)^{t-1},\]
    as there are $\kappa^{2(t-1)}$ choices for $x_0, \ldots, x_{t-2}$ given $x_{t-1}, C_1, \ldots, C_{t-1}$, and $\Delta$ choices for $e_0$ given $x_0$.
    As the above is independent of $C_1, \ldots, C_{t-1}$, we have
    \begin{align*}
        &~\Pr[d_i(v) = 1 \text{ and we reach Step~\ref{step: flag at fail}} \mid \phi_i] \\
        &\leq \frac{\Delta}{2(m - i + 1)}\sum_{t = 1}^T\left(\frac{\kappa^2}{\ell}\right)^{t-1}\sum_{u \in N[v]}\sum_{w \in N[u]}\Pr[M(\psi_{t - 1}, w) \cap C_t = \0 \mid \phi_i,\,x_{t-1} = u].
    \end{align*}

    The following claim will assist with the proof.

    \begin{claim}\label{claim: always missing}
        $\Pr[M(\psi_{t - 1}, w) \cap C_t = \0\mid \phi_i,\,x_{t-1} = u] \leq \frac{1}{\poly(\Delta)}$.
    \end{claim}
    \begin{claimproof}
        Note that $M(\psi_{t - 1}, w) \cap C_t$ is independent of $\phi_i$ and $x_{t-1}$ given $\psi_{t-1},C_1, \ldots, C_{t-1}$.
        Therefore, it is enough to show that
        \[\Pr[M(\psi_{t - 1}, w) \cap C_t = \0\mid \psi_{t-1},C_1, \ldots, C_{t-1}] \,\leq\, \frac{1}{\poly(\Delta)},\]
        as the right hand side above is independent of $\psi_{t-1},C_1. \ldots, C_{t-1}$.
        To this end, we note that $C_t$ is formed by sampling $\kappa$ times with replacement from the set $Q\defeq [(1+\eps/2)\Delta]\setminus \left(\cup_{j = 1}^{t-1}C_j\right)$.
        By our choice for $\kappa$, and the lower bound on $\Delta$ in the statement of Theorem~\ref{theo:main_theo}, we have
        \[|Q| \,\geq\, (1+\eps/2)\Delta - \kappa\,T \,\geq\, (1+\eps/100)\Delta.\]
        It follows that $|M(\psi_{t - 1},w) \cap Q| \,\geq\, \eps\Delta/100$.
        With this in hand, since we sample $\kappa$ times with replacement at Step~\ref{step:random_sample}, we have
        \begin{align*}
            \Pr[M(\psi_{t - 1}, w) \cap C_t = \0\mid \psi_{t-1},C_1, \ldots, C_{t-1}] &= \left(1 - \frac{|M(\psi_{t - 1}, w) \cap Q|}{|Q|}\right)^\kappa \\
            &\leq \exp\left(-\frac{\kappa\,\eps}{100(1 + \eps/2)}\right) \\
            &\leq \frac{1}{\poly(\Delta)},
        \end{align*}
        where the last step follows by our choice of $\kappa$, since $|Q| \leq (1+\eps/2)\Delta$, and since $\eps < 1$.
    \end{claimproof}

    By Claim~\ref{claim: always missing}, we have
    \begin{align}
        \Pr[d_i(v) = 1 \text{ and we reach Step~\ref{step: flag at fail}} \mid \phi_i] &\leq \frac{\Delta}{2(m - i + 1)}\sum_{t = 1}^T\left(\frac{\kappa^2}{\ell}\right)^{t-1}\sum_{u \in N[v]}\sum_{w \in N[u]}\frac{1}{\poly(\Delta)} \nonumber \\
        &\leq \frac{4\,\Delta^3}{(m - i + 1)\poly(\Delta)} \nonumber \\
        &\leq \frac{1}{(m - i + 1)\poly(\Delta)}, \label{eq: fail flag}
    \end{align}
    for sufficiently large constants in the definitions of $\kappa$ and $\ell$.
    Putting together \eqref{eq: fail end} and \eqref{eq: fail flag} completes the proof.
\end{proof}

Note that $d_i(v)$ is independent of all $d_j(v)$ for $j < i$, given $\phi_i$.
In particular, as a result of Lemma~\ref{lemma: degree of v in G star} we have
\begin{align*}
    \E[d_i(v) \mid d_1(v), \ldots, d_{i - 1}(v)] &= \E[\E[d_i(v) \mid d_1(v), \ldots, d_{i - 1}(v), \phi_i]] \\
    &= \E[\E[d_i(v) \mid \phi_i]] \\
    &= \E[\Pr[d_i(v) = 1 \mid \phi_i]] \\
    &\leq \frac{1}{(m - i + 1)\poly(\Delta)}.
\end{align*}

To finish the analysis of $d(v)$, we shall apply the following concentration inequality due to Kuszmaul and Qi \cite{azuma}, which is a special case of their version of multiplicative Azuma's inequality for supermartingales:

\begin{theorem}[{Kuszmaul--Qi \cite[Corollary 6]{azuma}}]\label{theo:azuma_supermartingale}
    Let $c > 0$ and let $X_1$, \ldots, $X_n$ be %real valued
    random variables taking values in $[0,c]$.
    Suppose that $\E[X_i|X_1, \ldots, X_{i-1}] \leq a_i$ for all $i$.
    Let $\mu \defeq \sum_{i = 1}^na_i$. Then, for any $\delta > 0$,
    \[\Pr\left[\sum_{i = 1}^nX_i \geq (1+\delta)\mu\right] \,\leq\, \exp\left(-\frac{\delta^2\mu}{(2+\delta)c}\right).\]
\end{theorem}

From our earlier computation, we may apply Theorem~\ref{theo:azuma_supermartingale} with 
\[X_i = d_i(v), \qquad c = 1, \qquad \text{and} \qquad a_i = \frac{1}{(m - i + 1)\poly(\Delta)}.\]
Assuming $G$ has no isolated vertices, we must have $n/2 \leq m \leq n^2$.
With this in hand, from the bounds on partial sums of the harmonic series \ep{see, e.g., \cite[\S2.4]{bressoud2022radical}}, we have $\mu = \dfrac{\log n}{\poly(\Delta)}$.
Setting $\delta = \dfrac{\eps\Delta}{10\mu}$, we conclude
\[\Pr\left[d(v) \geq \frac{\eps\Delta}{6}\right] \,\leq\, \Pr\left[d(v) \geq (1+\delta)\mu\right] \,\leq\, \exp\left(-\Omega(\eps\Delta)\right) \,=\, \frac{1}{\poly(n)},\]
where the last step follows by the lower bound on $\Delta$ stated in Theorem~\ref{theo:main_theo}.
By taking a union bound over $V(G)$, we have
\[\Pr[\text{Failure}] \leq \frac{1}{\poly(n)},\]
completing the proof of Proposition~\ref{prop:failure}.